\documentclass[twocolumn,noversion,nonote,nopagenumbers]{cdcarticle}

\usepackage{comment}
\usepackage{amsmath,amsthm,amssymb}
\usepackage{tikz}
\usepackage{graphicx}
\usepackage{enumerate}
\usepackage[hidelinks]{hyperref}
\usepackage{algorithm}
\usepackage{algpseudocode}

\usetikzlibrary{matrix,fit,backgrounds,plotmarks,calc}
\usepackage{pgfplots}

\pgfplotsset{width=9.1cm,height=8.2cm,compat=newest}

\tikzset{dashdot/.style={dash pattern=on 1pt off 2.5pt on 4.5pt off 2.5pt}}
\pgfplotscreateplotcyclelist{gray marks}{%
every mark/.append style={fill=gray,scale=0.8},mark=*\\%
every mark/.append style={fill=gray},mark=diamond*\\%
every mark/.append style={fill=gray},mark=triangle*\\%
every mark/.append style={fill=gray,scale=0.8},mark=square*\\%
}
\pgfplotscreateplotcyclelist{color dashes}{%
blue,very thick\\
red,very thick\\
blue,very thick,dashdot\\
red,very thick,dashdot\\
blue,very thick,dashed\\
red,very thick,dashed\\
}

\pgfplotsset{
  tick label style = {font=\footnotesize},
  every axis label = {font=\footnotesize},
  legend style = {font=\footnotesize},
  label style = {font=\normalsize}
}

\newtheorem{thm}{Theorem}

\newtheorem{lem}[thm]{Lemma}

\newtheorem{defn}[thm]{Definition}
\newtheorem{ex}[thm]{Example}
\newtheorem{cor}[thm]{Corollary}
\renewenvironment{proof}{{\noindent\bf Proof.}}{ \hfill ~\qed}
\def\qed{\rule[0pt]{5pt}{5pt}\par\medskip}

\newcommand{\bmat}[1]{\begin{bmatrix}#1\end{bmatrix}}		

\newcommand{\bmtx}{\begin{bmatrix}}
\newcommand{\emtx}{\end{bmatrix}}
\newcommand{\bsmtx}{\left[ \begin{smallmatrix}} 
\newcommand{\esmtx}{\end{smallmatrix} \right]}
\newcommand{\bmatarray}[1]{\left[\begin{array}{#1}}
\newcommand{\ematarray}{\end{array}\right]} 
\newcommand{\field}[1]{\mathbb{#1}}
\newcommand{\R}{\field{R}}
\newcommand{\Sm}{\field{S}}
\newcommand{\Ltwo}{\mathcal{L}_2[0,T]}
\newcommand{\RH}{\field{RH}_\infty}
\newcommand{\RL}{\field{RL}_\infty}

\begin{document}

\title{Finite Horizon Robustness Analysis of LTV Systems
 \\ Using Integral Quadratic Constraints}

	\author{Pete Seiler\footnotemark[1] 
          \and Robert M. Moore\footnotemark[2] 
	\and Chris Meissen\footnotemark[2]
	\and Murat Arcak \footnotemark[3]
	\and Andrew Packard\footnotemark[2]}

\note{}
\maketitle

\footnotetext[1]{P. Seiler is with the Department of AEM at the
  University of Minnesota, MN, USA. \texttt{seile017@umn.edu}}
\footnotetext[2]{R. M.~Moore, C.~Meissen, and
  A.~Packard are with the Department of Mechanical Engineering at the
  University of California, Berkeley, CA~94720, USA.\\
  \texttt{\{max.moore,cmeissen,apackard\}@berkeley.edu}}
\footnotetext[3]{M. Arcak is with the Department of Electrical
  Engineering and Computer Science at the University of California,
  Berkeley, CA~94720, USA. \texttt{arcak@berkeley.edu}} 

\begin{abstract}

  The goal of this paper is to assess the robustness of an uncertain
  linear time-varying (LTV) system on a finite time horizon.  The
  uncertain system is modeled as a connection of a known LTV system
  and a perturbation. The input/output behavior of the perturbation is
  described by time-domain, integral quadratic constraints (IQCs).
  Typical notions of robustness, e.g. nominal stability and gain/phase
  margins, can be insufficient for finite-horizon analysis.  Instead,
  this paper focuses on robust induced gains and bounds on the
  reachable set of states.  Sufficient conditions to compute robust
  performance bounds are formulated using dissipation inequalities and
  IQCs.  The analysis conditions are provided in two equivalent forms
  as Riccati differential equations and differential linear matrix
  inequalities.  A computational approach is provided that leverages
  both forms of the analysis conditions. The approach is demonstrated
  with two examples.

\end{abstract}

\section{Introduction}

This paper develops theoretical and computational methods to analyze
the robustness of linear time-varying (LTV) systems over finite time
horizons.  Motivating applications for this work include robotic
systems \cite{murray94,spong05} and space launch / re-entry vehicles
\cite{marcos09,veenman09} both of which undergo finite-time
trajectories.  Typical notions of robustness, e.g. nominal stability
and gain/phase margins, can be insufficient for such systems.  For
example, one approach is to evaluate the stability and performance of
the LTV system at ``frozen'' time instances along the
trajectory. However, there are LTV systems $\dot{x}(t)=A(t) x(t)$ for
which $A(t)$ is stable for each frozen time $t$ but with trajectories
that grow unbounded \cite{khalil01}.

This paper moves beyond the frozen analysis technique and instead
evaluates time-domain metrics over the finite horizon.  The analysis
is performed on an uncertain LTV system modeled, as shown in
Figure~\ref{fig:Gunc}, by an interconnection of a known, nominal LTV
system $G$ and a perturbation $\Delta$.  The perturbation is used to
model difficult to analyze elements including nonlinearities and
dynamic or parametric uncertainty.  The input-output properties of
$\Delta$ are characterized by integral quadratic constraints (IQCs)
\cite{megretski97}.  An extensive list of IQCs for various classes of
perturbations is given in \cite{megretski97,veenman16}.  The main
result in \cite{megretski97} is an (infinite-horizon), input-output
$\mathcal{L}_2$ stability theorem using frequency domain IQCs. The
proof relies on an operator theoretic, homotopy method.

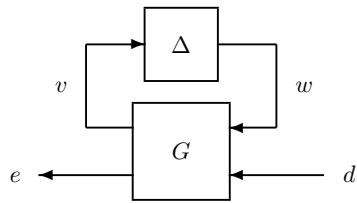
\begin{figure}[h]
\centering
\scalebox{0.9}{
\begin{picture}(172,90)(23,20)
 \thicklines
 \put(75,25){\framebox(40,40){$G$}}
 \put(163,32){$d$}
 \put(155,35){\vector(-1,0){40}}  
 \put(23,32){$e$}
 \put(75,35){\vector(-1,0){40}}  
 \put(80,75){\framebox(30,30){$\Delta$}}
 \put(42,70){$v$}
 \put(55,55){\line(1,0){20}}  
 \put(55,55){\line(0,1){35}}  
 \put(55,90){\vector(1,0){25}}  
 \put(143,70){$w$}
 \put(135,90){\line(-1,0){25}}  
 \put(135,55){\line(0,1){35}}  
 \put(135,55){\vector(-1,0){20}}  
\end{picture}
} 
\caption{Interconnection $F_u(G,\Delta)$ of a nominal LTV system $G$
  and perturbation $\Delta$.}
\label{fig:Gunc}
\end{figure}

The main contribution of this paper is an algorithm to compute
robustness metrics on a finite horizon. First, nominal finite-horizon
LTV performance is reviewed (Section~\ref{sec:nomperf}) focusing on
induced $\mathcal{L}_2$ and $\mathcal{L}_2$-to-Euclidean gains. The
$\mathcal{L}_2$-to-Euclidean gain is useful for computing bounds on
the reachable set of states.  Next, sufficient conditions are given
(Section~\ref{sec:robperf}) to bound the induced $\mathcal{L}_2$ and
$\mathcal{L}_2$-to-Euclidean gains for uncertain LTV systems. The
analysis is formulated with dissipation inequalities and IQCs.  This
yields performance conditions in the form of infinite-dimensional
differential linear matrix inequalities (DLMIs).  These conditions can
be equivalently re-written with a Riccati Differential Equation
(RDE). This equivalence is based on a variation of the strict bounded
real lemma (Theorem~\ref{thm:BRL} in Section~\ref{sec:nomperf}) which
generalizes existing results in
\cite{tadmor90,ravi91,green95,chen00}. The algorithm to compute finite
horizon robustness metrics (Section~\ref{sec:comp}) leverages both
forms of these conditions. The proposed approach is demonstrated with
two examples (Section~\ref{sec:ex}).

This paper adds to existing results to assess robustness of
time-varying systems. The most closely related work is
\cite{jonsson02} which also considers finite horizon robustness with
IQCs.  However, the work in \cite{jonsson02} does not consider
disturbances and the theoretical approach / resulting numerical
algorithm are different than given here.  The work in
\cite{pfifer16IJRNC} and \cite{fry17} is also relevant. Uncertain
linear parameter varying systems are considered in
\cite{pfifer16IJRNC}.  The analysis is formulated using dissipation
inequalities and IQCs. A similar approach is used in \cite{fry17} to
provide robustness conditions for discrete-time LTV systems.  The
theoretical and computational approaches provided here differ from
these previous works in order to handle continuous-time LTV systems on
finite horizons.  Other related work includes robustness analysis for
finite horizon batch processes \cite{ma01}, nonlinear systems
\cite{tierno97}, and periodic LTV systems via time-domain lifting
\cite{dullerud96,ma02,kim06} or harmonic balance / frequency domain
lifting \cite{wereley90,fardad08,shafi13}.

\textbf{Notation:} Let $\mathbb{R}^{n \times m}$ and $\mathbb{S}^{n}$
denote the sets of $n$-by-$m$ real matrices and $n$-by-$n$ real,
symmetric matrices. The finite-horizon $\Ltwo$ norm of a signal
$v:[0,T] \rightarrow \R^n$ is defined as
$\|v\|_{2, [0,T]} := \left( \int_0^T v(t)^T v(t) dt \right)^{1/2}$.
If $\|v\|_{2, [0,T]}$ is finite then $v \in \Ltwo$.  $\RL$ denotes the
set of rational functions with real coefficients that have no poles on
the imaginary axis.  $\RH$ is the subset of functions in $\RL$ that
are analytic in the closed right-half of the complex plane.  Finally,
let $E \in \Sm^n$ and $\alpha \in \R$ be given with $E\ge 0$ and
$\alpha>0$.  Define an ellipsoid in terms of $(E,\alpha)$ as
$\mathcal{E}(E,\alpha) := \{x \in \mathbb{R}^n \ | \ x^T E x \leq
\alpha^2 \}$

\section{Nominal Performance}
\label{sec:nomperf}

\subsection{Finite Horizon LTV Systems}

Consider an LTV system $G$ defined on $[0,T]$:
\begin{align}
  \label{eq:LTV1}
  \dot{x}(t) & = A(t) x(t) + B(t) d(t) \\
  \label{eq:LTV2}
  e(t) & = C(t) x(t) + D(t) d(t) 
\end{align}
$x\in \R^{n_x}$ is the state, $d \in \R^{n_d}$ is the input, and
$e \in \R^{n_e}$ is the output.  The state matrices
$A:[0,T] \rightarrow \R^{n_x \times n_x}$,
$B:[0,T] \rightarrow \R^{n_x \times n_d}$,
$C:[0,T] \rightarrow \R^{n_e \times n_x}$, and
$D:[0,T] \rightarrow \R^{n_e \times n_d}$ are piecewise-continuous
(bounded) functions of time.  It is assumed throughout that
$T<\infty$. Thus $d\in\Ltwo$ implies $x$ and $e$ are in
$\Ltwo$ for any $x(0)$ \cite{brockett15}.

Many different performance metrics can be defined for this (nominal)
finite-horizon LTV system.  This paper mainly focuses on two
specific metrics.  First, the {\it finite-horizon induced
  $\mathcal{L}_2$-gain} of $G$ is
\begin{align*}
\|G\|_{2,[0,T]} :=  \sup \bigg\{\frac{\|e\|_{2,[0,T]}} {\|d\|_{2,[0,T]}} \ 
     \bigg|  \  x(0)=0, 0 \not = d \in \mathcal{L}_{2, [0,T]} \bigg\}.
\end{align*}
As noted above, $d\in \Ltwo$ implies $e \in \Ltwo$.  Thus the
$\mathcal{L}_2$ gain is finite for any fixed horizon $T<\infty$.

Next, assume $D(T)=0$.  Then the {\it finite-horizon
  $\mathcal{L}_2$-to-Euclidean gain} of $G$ is
\begin{align*}
\|G\|_{E,[0,T]}:= \sup \bigg\{\frac{\|e(T)\|_2} {\|d\|_{2,[0,T]}} \ 
  \bigg|  \ x(0)=0, 0 \not = d \in \mathcal{L}_{2, [0,T]}  \bigg\}.
\end{align*}
The $\mathcal{L}_2$-to-Euclidean gain depends on the system output $e$
only at the final time $T$.  The assumption that $D(T)=0$ ensures this
gain is well-defined. 

The $\mathcal{L}_2$-to-Euclidean gain can be used to bound the set of
states $x(T)$ reachable by disturbances of a given norm.  This
\emph{reachable set} is formally defined as follows:
\begin{align*}
\mathcal{R}_\beta := \left\{ x(T) \ \big| \ x(0)=0, 
   \| d\|_{2,[0,T]}  \le \beta \right\}.
\end{align*}
If $C(T)=I_{n_x}$ and $D(T)=0$ then $e(T)=x(T)$.  In this special
case, if $\|G\|_{E,[0,T]}\le \gamma$ then
$\|x(T)\|_2 \le \gamma \|d\|_{2,[0,T]}$.  This implies the reachable
set $\mathcal{R}_\beta$ is contained in a sphere of radius
$\gamma \beta$.  More general ellipsoidal bounds on the reachable set
can be obtained by proper selection of the output matrices.  For
example, select $C:= E^{\frac{1}{2}}$ and $D:=0$ for some given
$E \in \Sm^{n_x}$ with $E \ge 0$. With these choices
$\|G\|_{E,[0,T]}\le \gamma$ implies an ellipsoidal bound on the
reachable set:
$\mathcal{R}_\beta \subseteq \mathcal{E}(E,\beta \gamma)$.  The size
of the reachable set scales with the norm of the disturbance input.
The state $x(t)$ at intermediate times $t\in [0,T]$ can similarly be
bounded using the $\mathcal{L}_2$-to-Euclidean gain $\|G\|_{E,[0,t]}$.



\subsection{Generic Quadratic Cost}
\label{sec:Jcost}

The two nominal performance metrics introduced above are generalized
in Section~\ref{sec:robperf} to assess robustness of uncertain
systems.  A generic quadratic cost function is defined next in order
to unify these various nominal and robust performance metrics.
Specifically, let $Q : [0,T] \rightarrow \Sm^{n_x}$,
$R : [0,T] \rightarrow \Sm^{n_d}$,
$S : [0,T] \rightarrow \R^{n_x \times n_d}$, and
$F \in \R^{n_x\times n_x}$ be given. $(Q,S,R)$ are assumed to be
piecewise continuous (bounded) functions. A quadratic cost function
$J: \Ltwo \rightarrow \R$ is defined by $(Q,S,R,F)$ as follows:
\begin{align}
  \nonumber
  J(d) & := x(T)^T F x(T)  + \int_0^T \bsmtx x(t) \\ d(t) \esmtx^T
  \bsmtx Q(t) & S(t) \\ S(t)^T & R(t) \esmtx 
  \bsmtx x(t) \\ d(t) \esmtx  dt \\
  \label{eq:J}
  & \mbox{subject to: Eq.~\ref{eq:LTV1} with } x(0)=0
\end{align}

The finite-horizon induced $\mathcal{L}_2$ gain of $G$ can be related
to the quadratic cost $J$ by proper choice of $(Q,S,R,F)$.  In
particular, let $\gamma>0$ be given and select $Q(t):= C(t)^TC(t)$,
$S(t):= C(t)^TD(t)$, $R(t):= D(t)^TD(t) - \gamma^2 I_{n_d}$, and
$F:=0$.  This yields the following cost function
\begin{align}
  J(d) = \| e\|_{2,[0,T]}^2 - \gamma^2 \| d\|_{2,[0,T]}^2
\end{align}
Thus $J(d) \le 0$ $\forall d \in \Ltwo$ if and only if
$\|G\|_{2,[0,T]} \le \gamma$.

The finite-horizon $\mathcal{L}_2$-to-Euclidean gain of $G$ can also
be related to the quadratic cost $J$ but with different choices for
$(Q,S,R,F)$.  Let $\gamma>0$ be given and select $Q(t):= 0$,
$S(t):= 0$, $R(t):= -\gamma^2 I_{n_d}$, and $F:=C^T(T)C(T)$.  This
yields the following cost function
\begin{align}
  J(d) = \| e(T) \|_2^2 - \gamma^2 \| d\|_{2,[0,T]}^2
\end{align}
Thus $J(d) \le 0$ $\forall d \in \Ltwo$ if and only if
$\|G\|_{E,[0,T]} \le \gamma$.

\subsection{Strict Bounded Real Lemma}


The next theorem states an equivalence between a bound on the
quadratic cost $J$ and the existence of a solution to a Riccati
Differential Equation (RDE) or Riccati Differential Inequality (RDI).
The theorem is expressed in terms of strict inequalities and
generalizes existing results for the induced $\mathcal{L}_2$ gain of
LTV systems \cite{tadmor90,ravi91,green95,chen00}.  



\begin{thm}
\label{thm:BRL}
Let $(Q,S,R,F)$ be given with $R(t) \prec 0$ for all $t \in [0,T]$.
The following statements are equivalent:
\begin{enumerate}
  \item $\exists \epsilon >0$ such that $J(d) \le -\epsilon \|d\|_{2,[0,T]}^2$
    $\forall d \in \Ltwo$.
  \item There exists a differentiable function
    $Y:[0,T] \rightarrow \Sm^n$ such that $Y(T)=F$ and
    \begin{align*}
      \dot{Y} +  A^T Y + YA +Q - (YB+S) R^{-1} (YB+S)^T = 0
    \end{align*}
    This is a Riccati Differential Equation (RDE).
  \item There exists $\epsilon >0$ and a differentiable function
    $P:[0,T] \rightarrow \Sm^n$ such that $P(T)\succeq F$ and
    \begin{align*}
      \dot{P} & +  A^T P + P A +Q \\
      &  - (PB+S) R^{-1} (PB+S)^T  \preceq -\epsilon I
    \end{align*}
    This is a strict Riccati Differential Inequality (RDI).
\end{enumerate}
\end{thm}
\begin{proof} 
  The proof of ($3 \Rightarrow 1$) is given below as it highlights the
  dissipation inequality framework.  The remainder of the proof is in
  Appendix~\ref{sec:BRLproof} for completeness.


  $\mathbf{(3 \Rightarrow 1)}$ By the Schur complement lemma \cite{boyd94},
  the RDI and $R(t)\prec 0$ imply $\exists \tilde{\epsilon}>0$ such
  that
  \begin{align}
    \label{eq:RDI2by2}
    \bmtx \dot{P}+A^T P+PA & PB \\ B^T P & 0 \emtx +
    \bmtx Q & S \\ S^T & R \emtx \preceq -\tilde\epsilon I
  \end{align}
  Next define a a quadratic storage function $V(x,t):=x^T P(t) x$.
  Let $x(t)$ be a solution of the LTV system (Equation~\ref{eq:LTV1})
  starting from $x(0)=0$ and forced by some input $d\in \Ltwo$.
  Multiply Equation~\ref{eq:RDI2by2} on the left and right by
  $[x(t)^T \, d(t)^T]$ and its transpose to obtain the following
  dissipation inequality:
  \begin{align}
    \dot{V} + \bmtx x \\ d \emtx^T
    \bmtx Q & S \\ S^T & R \emtx \bmtx x \\ d \emtx
    \le -\tilde\epsilon \bmtx x \\ d \emtx^T  \bmtx x \\ d \emtx
  \end{align}
  Integrate the dissipation inequality from $t=0$ to $t=T$:
  \begin{align*}
    & V(x(T),T) - V(x(0),0)  \\
    & + \int_0^T \bsmtx x(t) \\ d(t) \esmtx^T
    \bsmtx Q(t) & S(t) \\ S(t)^T & R(t) \esmtx \bsmtx x(t) \\ d(t) \esmtx
    \, dt
    \le -\tilde \epsilon \left\| \bsmtx x \\ d \esmtx \right\|_{2,[0,T]}^2
  \end{align*}
  Apply the boundary condition $P(T)\succeq F$ to obtain:
  \begin{align}
    \label{eq:JbndWithIC}
    J(d) \le V(x(0),0)-\tilde\epsilon \|d\|_{2,[0,T]}^2
  \end{align}
  This bound is valid for any $d\in L_2[0,T]$. Hence, applying
  $x(0)=0$ yields $J(d) \le -\tilde\epsilon \|d\|_{2,[0,T]}^2$
  $\forall d \in \Ltwo$.
\end{proof}
\vspace{0.05in}


This theorem assumes zero initial conditions $x(0)=0$.  This implies
that the initial stored energy $V(x(0),0)$ is zero and hence is
dropped from Equation~\ref{eq:JbndWithIC} in the proof. Non-zero
initial conditions can be incorporated by retaining this initial
stored energy in the performance bound.

Nominal performance is most easily assessed using the RDE.  The
performance $J(d) \le -\epsilon \|d\|_{2,[0,T]}^2$ is
achieved if the associated RDE exists on $[0,T]$ when integrated
backward from $Y(T)=F$.  The assumption $R(t)\prec 0$ ensures $R(t)$
is invertible and hence the RDE is well-defined for all $t\in [0,T]$.
Thus the solution of the RDE exists on $[0,T]$ unless it grows
unbounded.  As a concrete example, it was noted in
Section~\ref{sec:Jcost} that
$J(d) = \| e\|_{2,[0,T]}^2 - \gamma^2 \| d\|_{2,[0,T]}^2$ for specific
choices of $(Q,S,R,F)$. The matrix $R$, and hence the RDE, depends on
the choice of $\gamma$. For a fixed $\gamma>0$, the performance
$\|G\|_{2,[0,T]} < \gamma$ is achieved if the associated RDE exists on
$[0,T]$ when integrated backward from $Y(T)=0$.  The smallest bound on
the induced $\mathcal{L}_2$ gain can be found via bisection. The RDI
will be used later to assess the robustness of uncertain LTV systems.

%

\section{Robust Performance}
\label{sec:robperf}

\subsection{Uncertain LTV Systems}


An uncertain, finite-horizon LTV system is given by the
interconnection $F_u(G,\Delta)$ of a nominal LTV system $G$
and a perturbation $\Delta$ as shown in Figure~\ref{fig:Gunc}. The LTV
system $G$ is described by the following state-space model:
\begin{equation}
  \label{eq:LTVnom}
  \begin{split}
    & \dot{x}_G(t) = A_G(t)\, x_G(t) + B_{G1}(t)\,w(t) +  B_{G2}(t)\, d(t) \\
    & v(t)=C_{G1}(t)\,x_G(t)+D_{G11}(t)\, w(t)+ D_{G12}(t) \,d(t)\\
    & e(t)=C_{G2}(t)\,x_G(t)+D_{G21}(t)\, w(t)+ D_{G22}(t) \,d(t) 
  \end{split}
\end{equation}
where $x_G \in \R^{n_G}$ is the state. The inputs are $w\in \R^{n_w}$
and $d\in \R^{n_d}$ while $v\in \R^{n_v}$ and $e\in \R^{n_e}$ are
outputs.  The state matrices are piecewise continuous (bounded)
functions of time with appropriate dimensions, e.g.
$A_G: [0,T] \rightarrow \R^{n_G \times n_G}$.  The perturbation
$\Delta:\mathbf{L}_2^{n_v}[0,T] \rightarrow \mathbf{L}_2^{n_w}[0,T]$
is a bounded, causal operator. Well-posedness of the interconnection
$F_u(G,\Delta)$ is defined as follows.

\vspace{0.05in}
\begin{defn}
\label{def:wellpose}
$F_u(G,\Delta)$ is \underline{well-posed} if for all
$x_G(0)\in\R^{n_G}$ and $d\in \mathbf{L}_2^{n_d}[0,T]$ there exists
unique solutions $x_G\in \mathbf{L}_2^{n_G}[0,T]$,
$v\in \mathbf{L}_2^{n_v}[0,T]$, $e \in \mathbf{L}_2^{n_e}[0,T]$, and
$w\in \mathbf{L}_2^{n_w}[0,T]$ satisfying Equation~\eqref{eq:LTVnom}
and $w=\Delta(v)$ with a causal dependence on $d$.
\end{defn}
\vspace{0.05in}



The perturbation $\Delta$ can have block-structure as is standard in
robust control modeling \cite{zhou96}.  It can include blocks that are
hard nonlinearities (e.g. saturations) and infinite dimensional
operators (e.g. time delays) in addition to true system
uncertainties. The term ``uncertainty'' is used for simplicity when
referring to $\Delta$.

\subsection{Integral Quadratic Constraints (IQCs) }

IQCs \cite{megretski97} are used to describe the input/output behavior
of $\Delta$. They can be formulated in either the frequency or time
domain. The time domain formulation is more useful for analysis of
uncertain time-varying systems.  This formulation is based on the
graphical interpretation in Figure~\ref{fig:iqcpsi}.  The inputs and
outputs of $\Delta$ are filtered through an LTI system $\Psi$ with
zero initial condition $x_\psi(0)=0$. The dynamics of $\Psi$ are given
as follows:
\begin{align}
\label{eq:Psi}
\begin{split}
  & \dot{x}_\psi(t) = A_\psi \, x_\psi(t) + B_{\psi 1} \, v(t) + B_{\psi 2} \, w(t) \\
  & z(t) = C_\psi \, x_\psi(t) + D_{\psi 1} \, v(t) + D_{\psi 2} \, w(t) 
\end{split}
\end{align}
where $x_\psi \in \R^{n_\psi}$ is the state.   A time domain IQC is an
inequality enforced on the output $z$ over a finite horizon. The
formal definition is given next.


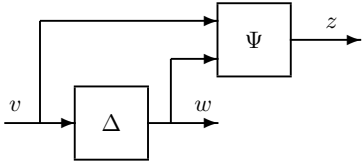
\begin{figure}[htbp]
\centering
\scalebox{0.9}{
\begin{picture}(160,70)(0,5)
\thicklines
\put(2,25){$v$}
\put(0,20){\vector(1,0){30}}
\put(30,5){\framebox(30,30){$\Delta$}}
\put(80,25){$w$}
\put(60,20){\vector(1,0){30}}
\put(15,20){\line(0,1){43}}
\put(15,63){\vector(1,0){75}}
\put(70,20){\line(0,1){27}}
\put(70,47){\vector(1,0){20}}
\put(90,40){\framebox(30,30){$\Psi$}}
\put(135,60){$z$}
\put(120,55){\vector(1,0){30}}
\end{picture}
} 
\caption{Graphical interpretation for time domain IQCs}
\label{fig:iqcpsi}
\end{figure}
\vspace{0.05in}

\begin{defn}
  \label{def:tdiqc}
  Let $\Psi \in \RH^{n_z \times (n_v+n_w)}$ and
  $M:[0,T] \rightarrow \Sm^{n_z}$ be given with $M$ piecewise
  continuous. A bounded, causal operator
  $\Delta:\mathbf{L}_2^{n_v}[0,T] \rightarrow \mathbf{L}_2^{n_w}[0,T]$
  satisfies the \underline{time domain IQC} defined by $(\Psi, M)$ if
  the following inequality holds for all
  $v \in \mathcal{L}_2^{n_v}[0,T]$ and $w=\Delta(v)$:
  \begin{align}
    \label{eq:tdiqc}
    \int_0^T z(t)^T M(t) z(t) \, dt \, \ge 0
  \end{align}
  where $z$ is the output of $\Psi$ driven by inputs $(v,w)$ with
  zero initial conditions $x_\psi(0)=0$.
\end{defn}
\vspace{0.05in}

The notation $\Delta \in \mathcal{I}(\Psi,M)$ is used when
$\Delta$ satisfies the corresponding IQC.  Time domain
IQCs, as defined above, are specified as finite-horizon constraints on
the outputs of $\Psi$. These are often referred to as hard IQCs
\cite{megretski97}. The definition given here only requires the IQC to
hold over the analysis horizon $[0,T]$. This is in contrast to hard
IQCs used for infinite horizon analysis which require the constraint
to hold over all finite time horizons.  Two examples of time domain
IQCs are provided below.

\begin{ex}
  \label{ex:LTIunc}
  Consider an LTI uncertainty $\Delta \in \RH$ with
  $\| \Delta \|_\infty \le 1$. Let $\Pi_{11} \in \RL$ be given with
  $\Pi_{11}(j\omega) = \Pi_{11}(j\omega)^* \ge 0$ for all
  $\omega \in \R \cup \{+\infty\}$. Then the following frequency
  domain IQC holds $\forall v \in L_{2}^{n_v}$ and $w=\Delta(v)$
  \begin{align}
    \label{eq:fdiqc}
    \int_{-\infty}^{\infty} \bsmtx V(j\omega) \\ W(j\omega) \esmtx^*
    \underbrace{\bsmtx \Pi_{11}(j\omega) & 0 \\ 0 & -\Pi_{11}(j\omega)
      \esmtx}_{:= \Pi(j\omega)} \bsmtx V(j\omega) \\ W(j\omega) \esmtx
    d\omega \ge 0
  \end{align}
  where $V$ and $W$ are Fourier transforms of $v$ and $w$. This IQC
  corresponds to the use of $D$-scales used in structured singular
  value $\mu$ analysis \cite{safonov80,doyle82,packard93,zhou96}.  
  A factorized representation for $\Pi$ yields a valid time domain
  IQC.  Specifically, let $\Pi = \Psi^\sim M \Psi$ where
  \begin{align}
    \label{eq:LTIuncIQC}
    \begin{split}
    \Psi & :=\bsmtx \Psi_{11} & 0 \\ 0 & \Psi_{11} \esmtx 
      \mbox{ with } \Psi \in \RH^{n_z \times 1} \\
    M & :=\bsmtx M_{11} & 0 \\ 0 & -M_{11} \esmtx
    \mbox{ with } M\in \Sm^{n_z} \mbox{ and } M_{11} \succeq 0
    \end{split}
  \end{align}
  It is shown in \cite{balakrishnan02} that $(\Psi,M)$ is a valid time
  domain IQC for $\Delta$ over any finite horizon $T<\infty$.
\end{ex}

\begin{ex}
  \label{ex:TVunc}
  Time domain IQCs are often specified with $\Psi$ as an LTI system
  and $M$ as a constant matrix.  Definition~\ref{def:tdiqc} above
  allows $M$ to be time-varying. This generalization is useful for
  time-varying real parameters. Let $\Delta:=\delta(t)$ where
  $\delta(t)\in \R$ and $|\delta(t)| \le 1$ for all $t \in [0,T]$.
  Define $\Psi:=I_2$ and
  $M(t):=\bsmtx m_{11}(t) & 0 \\ 0 & -m_{11}(t) \esmtx$ where
  $m_{11}:[0,T]\rightarrow \R$ is piecewise continuous and satisfies
  $m_{11}(t)\ge 0$.  Then $\Delta$ satisfies the time domain IQC
  defined by $(\Psi,M)$.  Time-varying IQCs can be defined for other
  uncertainties, e.g. see related work in \cite{pfifer15}.
\end{ex}

An extensive library of IQCs is provided in~\cite{megretski97} for
various types of perturbations. Most IQCs are specified in the
frequency domain using a multiplier $\Pi$.  Under some mild
assumptions, a valid time-domain IQC $(\Psi,M)$ can be constructed
from $\Pi$ via a $J$-spectral factorization \cite{seiler15}. This
allows the library of known (frequency domain) IQCs to be used for
time-domain, finite-horizon analysis.  More general IQC
parameterizations are not necessarily ``hard'' but can be handled with
the method in \cite{fetzer17}.


\subsection{Robust Induced $\mathcal{L}_2$ Gain}

The robustness of $F_u(G,\Delta)$ is analyzed using the
interconnection shown in Figure~\ref{fig:IQCAugment}.  The extended
system of $G$ (Equation~\ref{eq:LTVnom}) and the IQC filter $\Psi$
(Equation~\ref{eq:Psi}) is governed by the following state space
model:
\begin{align}
\label{eq:extsys}
\begin{split}
\dot{x}(t) & = \mathcal{A}(t) \, x(t) + \mathcal{B}(t) \,
     \bsmtx w(t) \\ d(t) \esmtx \\
z(t) & = \mathcal{C}_1(t) \, x(t) + \mathcal{D}_1(t) \,
     \bsmtx w(t) \\ d(t) \esmtx \\
e(t) & = \mathcal{C}_2(t) \, x(t) + \mathcal{D}_2(t) \,
     \bsmtx w(t) \\ d(t) \esmtx 
\end{split}
\end{align}
The extended state vector is $x:=\bsmtx x_G \\ x_\psi \esmtx \in \R^n$
where $n:=n_G+n_\psi$.  The state-space matrices are given by
(dropping the dependence on time $t$):
\begin{align*}
& \mathcal{A} := \bmtx A_G & 0 \\  B_{\psi 1} C_{G1} & A_{\psi} \emtx, 
\, \mathcal{B}:= \bmtx B_{G1} & B_{G2} \\ 
   B_{\psi 1} D_{G11} + B_{\psi 2} & B_{\psi 1} D_{G 12} \emtx \\
& \mathcal{C}_1 := \bmtx D_{\psi 1} C_{G1} & C_{\psi} \emtx, 
\, \mathcal{C}_2 := \bmtx C_{G2}\,\, & 0 \emtx, \\
& \mathcal{D}_1 := \bmtx D_{\psi 1} D_{G11}  + D_{\psi 2} & 
    D_{\psi 1} D_{G12} \emtx \\
& \mathcal{D}_2 := \bmtx D_{G21} & D_{G22} \emtx
\end{align*}

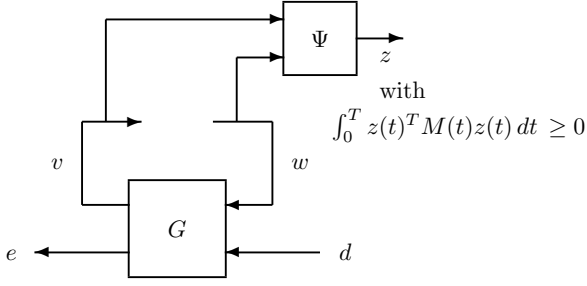
\begin{figure}[h]
\centering
\scalebox{0.9}{
\begin{picture}(220,120)(23,20)
 \thicklines
 \put(75,25){\framebox(40,40){$G$}}
 \put(163,32){$d$}
 \put(155,35){\vector(-1,0){40}}  
 \put(23,32){$e$}
 \put(75,35){\vector(-1,0){40}}  
 \put(42,70){$v$}
 \put(55,55){\line(1,0){20}}  
 \put(55,55){\line(0,1){35}}  
 \put(55,90){\vector(1,0){25}}  
 \put(143,70){$w$}
 \put(135,90){\line(-1,0){25}}  
 \put(135,55){\line(0,1){35}}  
 \put(135,55){\vector(-1,0){20}}  
 \put(65,90){\line(0,1){43}}
 \put(65,133){\vector(1,0){75}}
 \put(120,90){\line(0,1){27}}
 \put(120,117){\vector(1,0){20}}
 \put(140,110){\framebox(30,30){$\Psi$}}
 \put(180,115){$z$}
 \put(180,100){with}
 \put(160,85){$\int_0^T z(t)^T M(t) z(t) \, dt \, \ge 0$}
 \put(170,125){\vector(1,0){20}}
\end{picture}
} 
\caption{Extended LTI system of $G$ and filter $\Psi$.}
\label{fig:IQCAugment}
\end{figure}

The actual system to be analyzed is $F_u(G,\Delta)$ with input $d$ and
initial condition $x_G(0)=x_{G,0}$.  The analysis is instead performed
with the extended LTV system (Equation~\ref{eq:extsys}) and the
constraint $\Delta \in \mathcal{I}(\Psi,M)$.  The constrained extended
system has inputs $(d,w)$ and initial condition
$x(0)=\bsmtx x_{G,0} \\ 0 \esmtx$.  The IQC implicitly constrains the
input $w$.  {\it The IQC covers $\Delta$ such that the constrained
  extended system without $\Delta$ includes all behaviors of the
  original system $F_u(G,\Delta)$.} 


The following differential matrix inequality is used to assess the
robust performance of $F_u(G,\Delta)$~\footnote{The notation
  $(\cdot)^T$ in \eqref{eq:Rob_LMI} corresponds to an omitted factor
  required to make the corresponding term symmetric.}:
\begin{align}
\label{eq:Rob_LMI}
\begin{split}
&\bmtx \dot{P}+\mathcal{A}^T P + P \mathcal{A}  &  P \mathcal{B}  \\
   \mathcal{B}^T P & 0  \emtx
+ \bmtx Q & S \\ S^T & R \emtx \\
& \hspace{0.2in}
+ (\cdot)^T M \bmtx \mathcal{C}_1 & \mathcal{D}_1 \emtx  
\preceq -\epsilon I 
\end{split}
\end{align}
This inequality depends on the extended system, IQC, and quadratic
cost $(Q,S,R,F)$. It is compactly denoted as
$DLMI_{Rob}(P,M,\gamma^2,t)\preceq -\epsilon I$.  This notation
emphasizes that the constraint is a differential linear matrix
inequality (DLMI) in $(P,M,\gamma^2)$ for fixed $(G,\Psi)$ and
$(Q,S,R,F)$. The dependence on $(G,\Psi)$ and $(Q,S,R,F)$ is not
explicitly denoted but will be clear from context.


The next theorem formulates a sufficient condition to bound the
(robust) induced $\mathcal{L}_2$ gain of $F_u(G,\Delta)$.  The proof
uses IQCs and a standard dissipation argument~\cite{schaft99,
  willems72a, willems72b, khalil01}.  For induced $\mathcal{L}_2$
gains the quadratic cost matrices are chosen as $F:=0$ and
\begin{align}
\label{eq:RobL2toL2_QSRF}
\begin{split}
& Q(t):= \mathcal{C}_2(t)^T \mathcal{C}_2(t),
\,\,\, S(t):= \mathcal{C}_2(t)^T \mathcal{D}_2(t) \\
& R(t):= \mathcal{D}_2(t)^T \mathcal{D}_2(t) 
  - \gamma^2 \bsmtx 0_{n_w} & 0 \\ 0 & I_{n_d} \esmtx
\end{split}
\end{align}

\begin{thm}
\label{thm:RobL2toL2}
Let $G$ be an LTV system defined by \eqref{eq:LTVnom} and
$\Delta : \mathcal{L}^{n_v}[0,T] \rightarrow \mathcal{L}^{n_w}[0,T]$
be a bounded, causal operator. Assume $F_u(G,\Delta)$ is well-posed
and $\Delta \in \mathcal{I}(\Psi,M)$.  If there exists $\epsilon>0$,
$\gamma>0$ and a differentiable function $P:[0,T] \rightarrow \Sm^n$
such that $P(T) \succeq F$ and
\begin{align}
  DLMI_{Rob}(P,M,\gamma^2,t) \preceq -\epsilon I 
   \,\,\,\, \forall t \in [0,T]
\end{align}
then $\| F_u(G,\Delta) \|_{2,[0,T]} < \gamma$.
\end{thm}
\begin{proof}
  Let $d \in \Ltwo$ and $x_{G}(0)=x_{G,0}$ be given.  By
  well-posedness, $F_u(G,\Delta)$ has a unique solution $(x_G,v,w,e)$.
  As noted above, the extended system and IQC ``cover'' this system.
  In particular, forcing $\Psi$ with $(v,w)$ from $x_{\psi}(0)=0$
  yields $(x_\psi,z)$.  Define $x := \bsmtx x_G \\ x_\psi \esmtx$.
  Then $(x,z,e)$ are a solution of the extended system
  \eqref{eq:extsys} with inputs $(w,d)$ and initial condition
  $x(0)=\bsmtx x_{G,0} \\ 0 \esmtx$.  Moreover, $z$ satisfies the the
  IQC defined by $(\Psi,M)$.


  Define a storage function by $V(x,t) := x^T P(t) x$.  Left and right
  multiply the DLMI \eqref{eq:Rob_LMI} by $[x^T, w^T, d^T]$ and its
  transpose to show that $V$ satisfies the following dissipation
  inequality for all $t\in [0,T]$:
  \begin{align}
    \dot{V} + 
    \bmtx x \\ \bsmtx w \\ d \esmtx \emtx^T
    \bmtx Q & S \\ S^T & R \emtx 
    \bmtx x \\ \bsmtx w \\ d \esmtx \emtx
    + z^T M z
    \le -\epsilon \, d^T d
  \end{align}
  Use the choices for $(Q,S,R)$ in \eqref{eq:RobL2toL2_QSRF} to
  rewrite the second term as $e^Te - \gamma^2 d^T d$.
  Integrate over $[0,T]$ to obtain:
  \begin{align*}
    & x(T)^T P(T)x(T) - x_{G,0}^T P_{11}(0) x_{G,0} 
     + \int_{0}^{T}z^T(t) M(t)z(t) dt \\
    & - (\gamma^2-\epsilon) \|d\|^2_{2,[0,T]} +
    \|e\|^2_{2,[0,T]} \leq 0.
  \end{align*}
  Apply $P(T)\succeq F=0$ and $\Delta \in \mathcal{I}(\Psi,M)$ to conclude:
  \begin{align}
    \label{eq:RobL2toL2_WithIC}
    \|e\|^2_{2,[0,T]} \leq x_{G,0}^T P_{11}(0) x_{G,0} 
             + (\gamma^2-\epsilon) \|d\|^2_{2,[0,T]}
  \end{align}
  Finally, if $x_G(0)=0$ then 
  $\|F_u(G,\Delta)\|_{2,[0,T]} < \gamma$.
\end{proof}
\vspace{0.05in}

The effect of non-zero initial conditions $x_{G,0} \ne 0$ is captured
in Equation~\ref{eq:RobL2toL2_WithIC}.  If $d\equiv 0$, this
simplifies to
$x_{G,0}^T P_{11}(0) x_{G,0} \ge \|e\|^2_{2,[0,T]} \ge 0$ implying
$P_{11}(0) \succeq 0$.  Furthermore, on $[\tau, T]$ this implies
$P_{11} (\tau) \succeq 0$ for all $\tau \in [0,T]$. The IQC is valid
only for $x_\psi(0)=0$ and hence $P(\tau) \succeq 0$ need not hold in
general.


\subsection{Robust $\mathcal{L}_2$-to-Euclidean Gain}


A similar theorem provides a bound on the $\mathcal{L}_2$-to-Euclidean
gain of $F_u(G,\Delta)$. This requires the additional assumptions that
$D_{G21}(T)=0$ and $D_{G22}(T)=0$ so that $\mathcal{D}_2(T)=0$. Hence
$e(T) = \mathcal{C}_2(T) x(T)$ and the gain from $d$ to $e(T)$ is
well-defined.  To assess the robust $\mathcal{L}_2$-to-Euclidean gain
define $(Q,S,R,F)$ as:

\begin{align}
\label{eq:RobL2toE_QSRF}
\begin{split}
& Q(t):= 0, \, S(t):= 0, \,
R(t):= - \gamma^2 \bsmtx 0_{n_w} & 0 \\ 0 & I_{n_d} \esmtx, \\
& F:= \mathcal{C}_2^T(T) \mathcal{C}_2(T) 
= \bsmtx C_{G2}(T)^T C_{G2}(T) & 0 \\ 0 & 0 \esmtx
\end{split}
\end{align}
With these choices for $(Q,S,R,F)$ the next theorem is a minor
adaptation of Theorem~\ref{thm:RobL2toL2} and the proof is omitted.

\begin{thm}
\label{thm:RobL2toE}
Let $G$ be an LTV system defined by \eqref{eq:LTVnom} with
$D_{G21}(T)=0$ and $D_{G22}(T)=0$. Let
$\Delta : \mathcal{L}^{n_v}[0,T] \rightarrow \mathcal{L}^{n_w}[0,T]$
be a bounded, causal operator. Assume $F_u(G,\Delta)$ is well-posed
and $\Delta \in \mathcal{I}(\Psi,M)$.  If there exists $\epsilon>0$,
$\gamma>0$, and a differentiable function $P:[0,T] \rightarrow \Sm^n$
and such that $P(T) \succeq F$ and
\begin{align}
  DLMI_{Rob}(P,M,\gamma^2,t) \preceq -\epsilon I 
   \,\,\,\, \forall t \in [0,T]
\end{align}
then $\|F_u(G,\Delta)\|_{E,[0,T]} < \gamma$.
\end{thm}


The condition in Theorem~\ref{thm:RobL2toE} robustly bounds the states
$x_G(T)$ reachable by disturbances for any uncertainty
$\Delta \in \mathcal{I}(\Psi,M)$.  Note
$\mathcal{C}_2(T):=\bmtx C_{G2}(T) & 0 \emtx$ so that $e(T)$ only
depends on $x_G(T)$.  The IQC filter $\Psi$ is used only for analysis
and $x_\psi(T)$ is neglected in the bound.


Robust reachable sets with non-zero initial conditions $x_G(0) \ne 0$
can be computed with minor modifications. For example, assume the
initial condition of $G$ lies within the ellipsoid
$x_G(0) \in \mathcal{E}(E_0,1)$ for some $E_0\succ 0$.  The IQC still
requires $x_\psi(0)=0$.  Next, enforce
$P(0) \preceq \alpha_1 \bsmtx E_0 & 0 \\ 0 & 0 \esmtx$ for some
$\alpha_1>0$ (in addition to the conditions in
Theorem~\ref{thm:RobL2toE}).  It follows from the dissipation
inequality proof that the terminal state of $G$ is bounded by
$\|C_{G2}(T) x_G(T)\|_2 < \alpha_1 + \gamma \|d\|_{2,[0,T]}$.
Additional variations on robust reachable sets with non-zero initial
conditions can be found in Chapter 2 of \cite{moore15}.


\subsection{RDE Condition for Robust Performance}

Theorems~\ref{thm:RobL2toL2} and \ref{thm:RobL2toE} provide a
DLMI~\eqref{eq:Rob_LMI} to bound the induced $\mathcal{L}_2$ and
$\mathcal{L}_2$-to-Euclidean gain of $F_u(G,\Delta)$.  More
general robust performance conditions can be formulated by proper
choice of $(Q,S,R,F)$.  The numerical algorithm proposed in
Section~\ref{sec:comp} relies on an equivalence between the
DLMI~\eqref{eq:Rob_LMI} and a related RDE condition.  This equivalence
is demonstrated with an extended quadratic cost function $\mathcal{J}$
that combines the performance specification $(Q,S,R,F)$ and the IQC
$(\Psi,M)$. Specifically, define $\mathcal{J}$ with the extended
dynamics in \eqref{eq:extsys}:
$\dot{x} = \mathcal{A} x + \mathcal{B} \bsmtx w \\ d \esmtx$.  The
cost matrices $(\mathcal{Q},\mathcal{S},\mathcal{R},\mathcal{F})$ are
chosen as:
\begin{align}
  \label{eq:RobPerf_QSR}
  \begin{split}
    & \bmtx \mathcal{Q} & \mathcal{S} \\
    \mathcal{S}^T & \mathcal{R} \emtx :=
    (\cdot)^T M \bmtx \mathcal{C}_1 & \mathcal{D}_1 \emtx 
    +     \bmtx Q & S \\ S^T & R \emtx \\
    & \mathcal{F} := F
  \end{split}
\end{align}
The quadratic cost associated with these choices is:
\begin{align*}
  \mathcal{J}\left( \bsmtx w \\ d \esmtx \right) & :=
  x(T)^T \mathcal{F} x(T) 
  + \int_{0}^{T}z^T(t) M(t)z(t) dt \\
  & 
  + \int_0^T  \bmtx x(t) \\ \bsmtx w(t) \\ d(t) \esmtx \emtx^T
    \bmtx Q(t) & S(t) \\ S(t)^T & R(t) \emtx                 
    \bmtx x(t) \\ \bsmtx w(t) \\ d(t) \esmtx \emtx \, dt 
\end{align*}


The next corollary states the equivalence between the DLMI and RDE
conditions.  The DLMI can be rewritten as an RDI by the Schur
complement lemma \cite{boyd94}.  Hence the corollary follows
directly from Theorem~\ref{thm:BRL}.

\begin{cor}
\label{cor:RobL2toL2_RDE}
Let $(\mathcal{Q},\mathcal{S},\mathcal{R},\mathcal{F})$
be given by \eqref{eq:RobPerf_QSR}.  The
following are equivalent for any $\epsilon>0$ and
$\gamma>0$:
\begin{enumerate}
\item There exists a differentiable
  function $P:[0,T] \rightarrow \Sm^n$ such that $P(T) \succeq F$
  and $DLMI_{Rob}(P,M,\gamma^2,t)\preceq -\epsilon I$.
\item $\mathcal{R}(t) \prec 0$ for all $t \in [0,T]$. In addition,
  there exists a differentiable function $Y:[0,T] \rightarrow \Sm^n$
  such that $Y(T)=F$ and
  \begin{align*}
    \dot{Y} +  \mathcal{A}^T Y + Y \mathcal{A} + \mathcal{Q}
    - (Y\mathcal{B}+\mathcal{S}) 
    \mathcal{R}^{-1} (Y\mathcal{B}+\mathcal{S})^T = 0
  \end{align*}
\end{enumerate}
\end{cor}

\section{Computational Approach}
\label{sec:comp}

This section describes computational details and presents an algorithm
that combines complementary aspects of the DLMI and RDE robust
performance conditions.

\subsection{IQC Parameterization}


There is typically an infinite set of valid IQCs for a given
uncertainty $\Delta$.  Numerical implementations using IQCs often
involve a fixed choice for $\Psi$ and optimization subject to convex
constraints on $M$ \cite{megretski97,veenman16,palframan17}. The
algorithms given in the following sections will use such
parameterizations.  Two examples are given below.

\begin{ex}
  \label{ex:LTIuncParam}
  Consider an LTI uncertainty $\Delta \in \RH$ with
  $\|\Delta \|_\infty \le 1$. By Example~\ref{ex:LTIunc}, $\Delta$
  satisfies any IQC $(\Psi,M)$ with
  $\Psi :=\bsmtx \Psi_{11} & 0 \\ 0 & \Psi_{11} \esmtx$, 
  $M :=\bsmtx M_{11} & 0 \\ 0 & -M_{11} \esmtx$, and
  $M_{11} \succeq 0$.  A typical choice for $\Psi_{11}$ is
  \cite{veenman16}:
  \begin{align}
    \Psi_{11}^v:= \left[1, \frac{1}{(s+p)},\ldots \frac{1}{(s+p)^v}  \right]^T
    \mbox{ with } p>0
  \end{align}  
  The robustness analysis is performed by selecting $(p,v)$ to obtain
  (fixed) $\Psi$ and optimizing over $M_{11} \succeq 0$.  The results
  depend on the choice of $(p,v)$. Larger values of $v$ represent
  a richer class of IQCs and hence yield less conservative results but
  with increasing computational cost.  Further details on this
  parameterization are given in \cite{veenman16}.
\end{ex}


\begin{ex}
  The analysis can incorporate conic combinations of multiple IQCs.
  Let $(\Psi_1,M_1)$ and $(\Psi_2,M_2)$ define valid IQCs for
  $\Delta$. Hence $\int_0^T z_i^T M_i z_i \, dt \ge 0$ where $z_i$ is
  the output $\Psi_i$ driven by $v$ and $w=\Delta(v)$.  For any
  $\lambda_1, \lambda_2 \ge 0$ the two constraints can be combined to
  yield:
  \begin{align}
    \int_0^T \lambda_1 z_1^T M_1 z_1 + \lambda_2 z_2^T M_2 z_2 \, dt \ge 0
  \end{align}
  Thus a valid time-domain IQC for $\Delta$ is given by 
  \begin{align}
    \Psi:=\bsmtx \Psi_1 \\ \Psi_2 \esmtx 
    \mbox{ and }
    M(\lambda):= \bsmtx \lambda_1 M_1 & 0 \\ 0 & \lambda_2 M_2 \esmtx
  \end{align}
  The analysis is performed by selecting $(\Psi_i,M_i)$
  and optimizing over $\lambda$.  
\end{ex}



\subsection{Analysis with the DLMI Condition}

Assume the IQC is $(\Psi,M)$ with $\Psi$ fixed and $M$ constrained to
lie within a feasible set $\mathcal{M}$ described by LMIs.  The DLMI
\eqref{eq:Rob_LMI} has the same form for induced $\mathcal{L}_2$ and
$\mathcal{L}_2$-to-Euclidean gains but with different choices of
$(Q,S,R,F)$. In both cases the DLMI is linear in $(P,M,\gamma^2)$ for
fixed $(G,\Psi)$. The dependence on $\gamma^2$ enters via $R$. The
best (smallest) bound on the robust gain can be computed from a convex
semidefinite program (SDP):
\begin{align*}
  & \min \gamma^2  \\
  & \mbox{subject to: } M\in \mathcal{M}, \, P(T) \succeq F \\
  &  DLMI_{Rob}(P,M,\gamma^2,t) \preceq -\epsilon I 
     \,\, \forall t \in [0,T]
\end{align*}
There are two main issues with solving this SDP.  First, the DLMI
corresponds to an infinite number of constraints since it must hold
for all $t \in [0,T]$. This can be approximated by enforcing the DLMI
on a finite time grid $t_{DLMI}:=\{t_k\}_{k=1}^{N_g} \subset [0,T]$. 


Second, the optimization requires a search over the space of functions
$P:[0,T] \rightarrow \Sm^n$. This issue is addressed by restricting
$P$ to be a linear combination of differentiable basis functions.
Specifically, let $h_j:[0,T] \rightarrow \R$ $(j=1,\ldots,N_s)$ and
$H_k:[0,T] \rightarrow \Sm^n$ $(k=1,\ldots,N_m)$ be given scalar
and matrix differentiable basis functions. The storage function and
its derivative are given by:
\begin{align}
    P(t) & = \sum_{j=1}^{N_s} h_j(t) X_j + \sum_{k=1}^{N_m} H_k(t) x_k \\
    \dot{P}(t) & = \sum_{j=1}^{N_s} \dot{f}_j(t) X_j 
              + \sum_{k=1}^{N_m} \dot{F}_k(t) x_k 
\end{align}
Here $\{ X_j \}_{j=1}^{N_s} \subset \Sm^n$ and
$\{ x_k \}_{k=1}^{N_m} \subset \R$ are optimization variables. Many
choices are possible for the basis functions.  Initial work in
\cite{moore15} used scalar basis functions generated with a cubic
spline and no matrix basis functions.  The spline is constructed by
selecting an interpolation time grid
$\tau_{sp}:=\{\tau_j\}_{j=1}^{N_s}$ where $\tau_j < \tau_{j+1}$.
Note, the spline grid $\tau_{sp}$ is distinct from the DLMI grid
$t_{DLMI}$.  The spline consists of $N_s-1$ cubic functions defined on
the intervals $[\tau_j,\tau_{j+1}]$.  It interpolates the decision
variables $\{ X_j \}_{j=1}^{N_s}$, i.e. $P(\tau_j) = X_j$.  The cubic
functions satisfy boundary conditions to ensure continuity of the
spline and its first/second derivatives at the interval endpoints.
The corresponding spline basis functions $\{h_j \}_{j=1}^{N_s}$ are
not easy to express in explicit form but they can be evaluated
numerically at any $t \in [0,T]$.  Additional details are given
in \cite{moore15}.  The algorithm proposed below also uses a matrix
basis function generated by the RDE condition.




The approximations for the DLMI and $P$ lead to a finite dimensional
SDP in variables $\{ X_j \}_{j=1}^{N_s}$, $\{ x_k \}_{k=1}^{N_m}$,
$M$, and $\gamma^2$.  The optimization can be performed with standard
SDP solvers.  Enforcing the DLMI only on a finite grid decreases the
optimal cost relative to the original infinite-dimensional SDP.
Conversely, restricting $P$ to lie in a finite dimensional subspace
increases the optimal cost.  The solution accuracy depends on the
choice for the constraint time grid and basis functions.  A denser
time grid and additional bases functions will improve the accuracy
but with increased computation time.

\subsection{Analysis with the RDE Condition}

The RDE conditions for (robust) induced $\mathcal{L}_2$ and
$\mathcal{L}_2$-to-Euclidean gains do not require the constraint and
basis function approximations needed for the corresponding DLMI.
Specifically for any $(M,\gamma^2)$ the RDE can be
integrated\footnote{It is still assumed that $(G,\Psi)$ are given and
  fixed.} within a specified numerical accuracy using standard ODE
solvers.  If the RDE exists on $[0,T]$ when integrated backward from
$Y(T)=F$ then the robust gain is less than $\gamma$.  Bisection on
$\gamma$ can be used to find the smallest bound on the robust gain.
The difficulty with the RDE condition is that IQC matrix $M$ enters in
a non-convex fashion.  In most cases it would be computationally
expensive to perform numerical gradient searches over $M$ to find the
smallest bound $\gamma$.

\subsection{Combined Algorithm}

Algorithm~\ref{alg:comb} combines the DLMI and RDE conditions.  The
plant $G$ and IQC filter $\Psi$ are given.  The algorithm is
initialized with a stopping tolerance $tol$, a max number of
iterations $N_{iter}$, a time grid $t_{DLMI}$ to enforce
the DLMI, a time grid $\tau_{sp}$ for the (scalar) spline basis
functions, and a single (zero) matrix basis function $H_1$.

The first step is to solve the finite SDP by enforcing the DLMI on
$t_{DLMI}$.  This returns, if feasible, $\gamma_{SDP}^{(1)}$,
$M^{(1)}$, and the storage function decision variables
$\{ X_j^{(1)} \}_{j=1}^{N_s}$.  The next step is to hold the IQC
matrix fixed at $M^{(1)}$ and bisect to find the smallest $\gamma$
such that the RDE solution exists on $[0,T]$.  This yields
$\gamma_{RDE}^{(1)}$, $P_{RDE}^{(1)}$, and $t_{RDE}^{(1)}$.  Here
$t_{RDE}^{(1)}$ denotes the (dense) grid of time points returned by
the ODE solver associated with $P_{RDE}^{(1)}$.  

Two updates are performed before the next iteration.  First, the
matrix basis function is set equal to the RDE solution if
$\gamma_{RDE}^{(1)}<\infty$.  This choice is optimal for $M^{(1)}$. At
the next iteration the cubic splines are essentially used to perturb
around $P_{RDE}^{(1)}$.  The second update involves the DLMI time
grid.  In particular, if the DLMI time grid is too coarse then
$\gamma_{SDP}^{(1)} < \gamma_{RDE}^{(1)}$. In this case the DLMI is
evaluated on the (dense) grid of time points $t_{RDE}^{(1)}$.  The
time points where the DLMI is infeasible (or some subset) are added to
$t_{DLMI}$.  The algorithm terminates if the RDE and SDP results are
close or the maximum number of iterations has been reached.  Otherwise
the subsequent iterations proceed in the same fashion.


\begin{algorithm}
  \linespread{1.15}\selectfont
  \caption{Combined DLMI/RDE Approach} \label{alg:comb}  
  \begin{algorithmic}[1]
    \State \textbf{Given:} $G$ and $\Psi$
    \State \textbf{Initialize:} $tol$, $N_{iter}$, $t_{DLMI}:=\{t_k\}_{k=1}^{N_g}$,
    $\tau_{sp}:=\{\tau_j\}_{j=1}^{N_s}$, and $H_1\equiv 0$.

    \For{$i=1:N_{iter}$} 

    \State \parbox[t]{0.9\linewidth}{ \textbf{Solve SDP:} Enforce
      DLMI on $t_{DLMI}$. Use spline basis functions defined by
      $\tau_{sp}$ and matrix basis function $H_1$.}
    
    \State  \textbf{Output:}  $\gamma_{SDP}^{(i)}$, $M^{(i)}$, and
    decision vars. for $P$.

    \State

    \State \parbox[t]{0.9\linewidth}{\textbf{Solve RDE:} Hold
      $M^{(i)}$ fixed and bisect to find smallest $\gamma$
      such that the RDE exists on $[0,T]$.
    }
    
    \State \textbf{Output:} $\gamma_{RDE}^{(i)}$, $P_{RDE}^{(i)}$, and
       $t_{RDE}^{(i)}$.

    \State

    \State \textbf{Updates:}
    \State If $\gamma_{RDE}^{(i)}<\infty$ then $H_1:=P_{RDE}^{(i)}$ 
        else $H_1:= 0$.
    \State Add time points to $t_{DLMI}$ 
                if $\gamma_{SDP}^{(i)} < \gamma_{RDE}^{(i)}$.

    \State

    \If{$|\gamma_{SDP}^{(i)} - \gamma_{RDE}^{(i)}| < tol \cdot
      \gamma_{SDP}^{(i)}$}      
    \State Terminate iteration
    \EndIf

    \EndFor
  \end{algorithmic}
\end{algorithm}


\section{Examples}
\label{sec:ex}

\subsection{Robust Induced $\mathcal{L}_2$ Gain}

Consider an uncertain system $F_u(G,\Delta)$ with $\Delta \in \RH$ and
$\|\Delta\|_\infty \le 1$. $G$ is an LTI system defined by:
\begin{align*}
A_G & := \bsmtx -0.8 &  -1.3 &  -2.1 &  -2.5 \\
         2 &  -0.9 &  -8.4 &   0.7 \\
         2 &   8.6 &  -0.5 &  12.5 \\
       2.1 &  -0.3 & -12.6 &  -0.6 \esmtx
&& B_G := \bsmtx -0.6 &  1 \\ 0 & 0.2 \\
       0 & 0.4 \\ -1.3 & -0.2 \esmtx \\
C_G & := \bsmtx  -1.4 &  0  & 0.5 &  0 \\
                  0 &  -0.1 &   1 &   0 \esmtx
&& D_G := \bsmtx -0.3 & 0 \\ 0 & 0 \esmtx
\end{align*}
The infinite-horizon, worst-case induced $\mathcal{L}_2$ gain is 1.49
as computed with the \texttt{wcgain} function in Matlab.
Finite-horizon robust gains are computed with
Algorithm~\ref{alg:comb}. The IQC parameterization in
Example~\ref{ex:LTIuncParam} is used with $v=1$ and $p=10$.
Algorithm~\ref{alg:comb} is initialized for each horizon $T$ with
$tol = 5 \times 10^{-3}$, $N_{iter} = 10$, $t_{DLMI}$ as 20 evenly
spaced points in $[0,T]$, and $\tau_{sp}$ as 10 evenly spaced points
in $[0,T]$.  Figure~\ref{fig:RobL2} shows the finite-horizon robust
gains (blue solid) for $T:=\{1, 2, 5, 10, 20, 30, 40, 50, 100\}$. The
red dashed line denotes the infinite-horizon robust gain of 1.49.  It
took 466 secs to compute all nine finite-horizon results on a standard
laptop.  The iteration for $T=5 sec$ terminated in 3 steps and all
other iterations terminated in 2 steps. Matlab's \texttt{LMILab} and
\texttt{ode45} were used to solve the SDP and integrate the RDE in
Algorithm~\ref{alg:comb}. The ODE options were set to have an absolute
and relative error of $10^{-8}$ and $10^{-5}$, respectively.


\begin{figure}[h!] 
  \centering
  \includegraphics[scale=0.55]{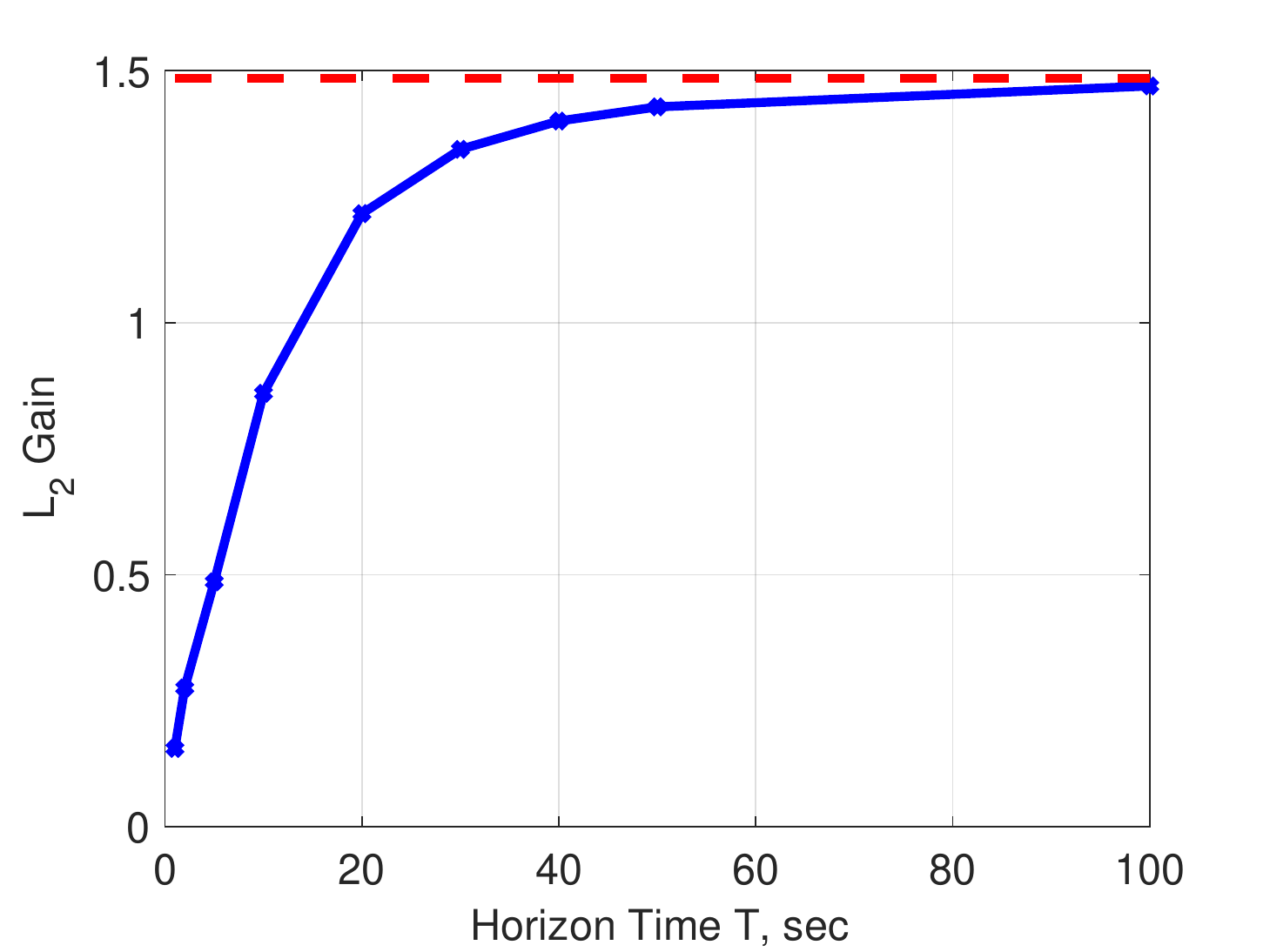}
  \caption{Robust Induced $\mathcal{L}_2$ Gain vs. Time Horizon (blue)
     and infinite horizon result (red dashed)}
  \label{fig:RobL2}
\end{figure}

\subsection{Two-link robot arm}

This example considers the robustness of a two link robot arm, shown
in Figure~\ref{fig:twoLinkRobot}, as it traverses a given finite-time
trajectory.  The mass and moment of inertia of the $i$-th link are
denoted by $m_i$ and $I_i$.  The robot properties are $m_1=3kg$,
$m_2 = 2kg$, $l_1 = l_2= 0.3m$, $r_1=r_2 = 0.15m$,
$I_1= 0.09 kg\cdot m^2$, and $I_2= 0.06 kg\cdot m^2$.  The equations
of motion \cite{murray94} for the two-link robot arm are given by:
\begin{align} 
\label{eq:linkDyns}
\begin{split}
&  \bmat{\alpha+ 2\beta\cos(\theta_2) & \delta + \beta \cos(\theta_2) \\
    \delta +  \beta \cos(\theta_2) & \delta}
  \bmat{\ddot{\theta}_1 \\ \ddot{\theta}_2} + \\
&  \bmat{-\beta \sin(\theta_2) \dot{\theta}_2 &
    -\beta \sin(\theta_2) (\dot{\theta}_1 + \dot{\theta}_2) \\
    \beta \sin(\theta_2) \dot{\theta}_1 & 0} \bmat{\dot{\theta}_1 \\
    \dot{\theta}_2} =\bmat{\tau_1 \\ \tau_2}
\end{split}
\end{align}
where $\tau_i$ is the torque applied to the base of the $i$-th link
and the model parameters are:
\begin{align*}
\alpha &:= I_{1} + I_{2} + m_1r_1^2 + m_2(l_1^2 + r_2^2) = 0.4425 \, kg \cdot m^2\\
\beta &:= m_2 l_1 r_2 = 0.09 \, kg \cdot m^2 \\
\delta &:= I_{2} + m_2 r_2^2 = 0.105 \, kg \cdot m^2
\end{align*}

\begin{figure}[h!] 
  \centering
  \includegraphics[scale=0.4]{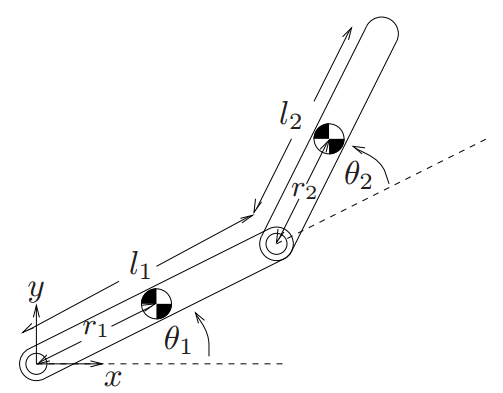}
  \caption{Two link robot arm \cite{murray94}.}
  \label{fig:twoLinkRobot}
\end{figure}

The state and input are denoted by
$\eta = \bmat{\theta_1 & \dot{\theta}_1 &\theta_2 & \dot{\theta_2}}^T$
and $\tau = \bmat{\tau_1 & \tau_2}^T$.  A trajectory $\bar{\eta}$ was
selected for the arm and the required input torque $\bar{\tau}$ was
computed.  Figure~\ref{fig:LinkRobotFigure2} shows the desired
trajectory for the tip of arm two (red dashed line) in Cartesian
coordinates from $t=0$ to $T=5$ sec.  The robot arm positions at four
different times are also shown.

\begin{figure}[h!] 
  \centering
  \includegraphics[scale=0.5]{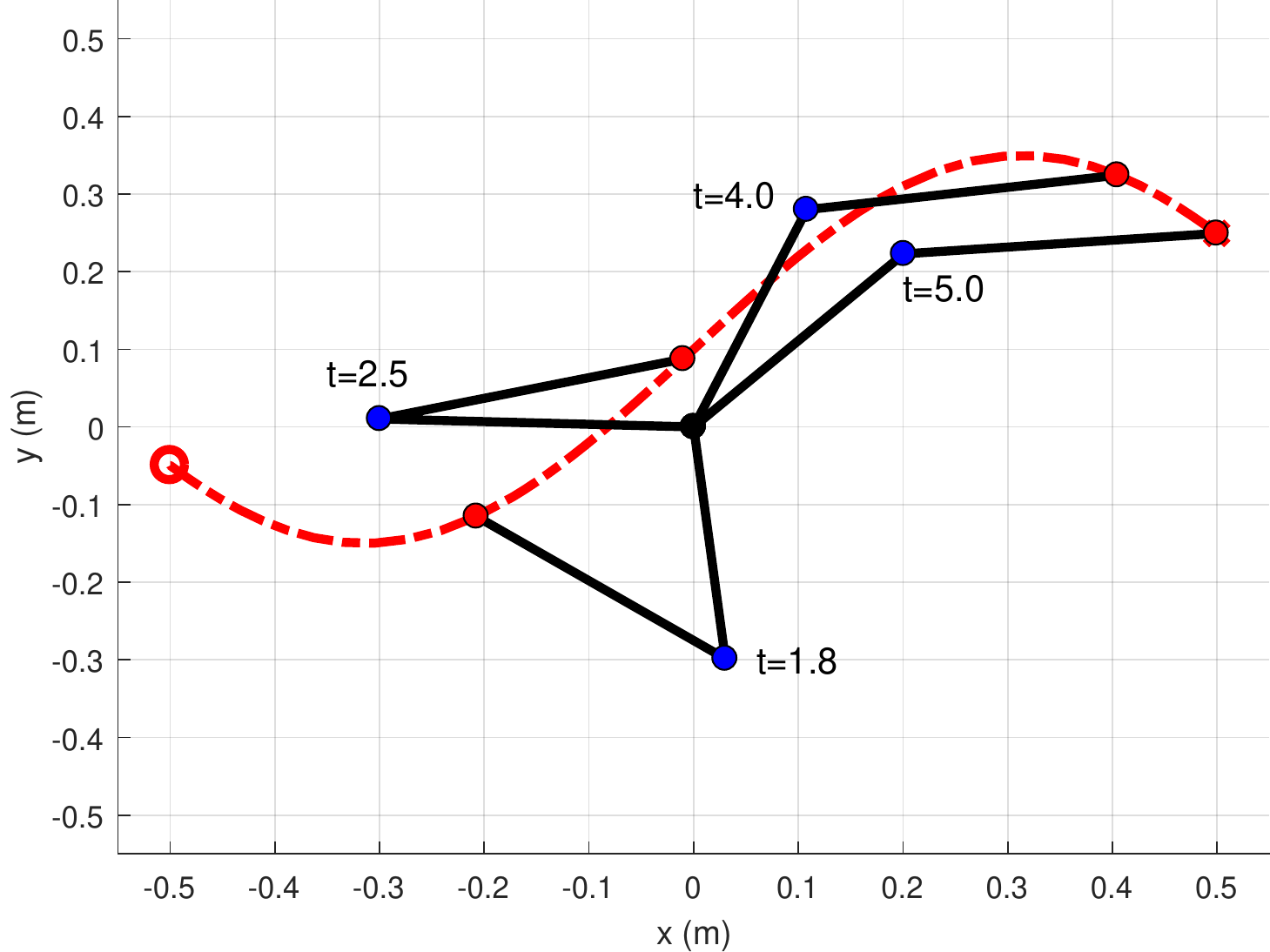}
  \caption{Desired trajectory in Cartesian coordinates
    (dotted red line) and robot arm position at four times.}
  \label{fig:LinkRobotFigure2}
\end{figure}

The objective is for the robot to track this trajectory in
the presence of small torque disturbances $d$. The input torque vector
is $\tau = \bar{\tau} + u + d$ where $u$ is an additional control
torque (specified below) to reject the disturbances. The
nonlinear dynamics \eqref{eq:linkDyns} are linearized around the
trajectory $(\bar{\eta}, \bar{\tau})$ to obtain an LTV
system $P$:
\begin{align}
  \dot{x}(t) = A(t) x(t) + B(t) \left( u(t) + d(t) \right)
\end{align}
where $x(t):=\eta(t)-\bar{\eta}(t)$ is the deviation of the nonlinear
state from the equilibrium trajectory.  The state matrices $(A,B)$
were computed at 200 uniformly spaced points in $[0,5]$. These state
matrices are linearly interpolated to obtain the LTV system at any
$t \in [0,T]$.

Next, a time-varying state feedback law $u(t)=-K(t)x(t)$ is designed
to improve the disturbance rejection.  The feedback gain is
constructed via finite horizon, LQR design with the following cost
function:
\begin{equation} \nonumber
J(x, u) = x(T)^TFx(T) 
   + \int_0^T \bsmtx x(t) \\ u(t) \esmtx^T
     \bsmtx Q & S  \\ S^T & R \esmtx
     \bsmtx x(t) \\ u(t) \esmtx \ dt
\end{equation}
where $Q:=\text{diag}(100,10,100,10)$, $R:= \text{diag}(0.1,0.1)$,
$S=0$ and $F:=\text{diag}(1,0.1,1,0.1)$. The
optimal feedback gain is $K(t)= R^{-1}B(t)^TP(t)$ where
$P:[0,T] \rightarrow \mathbb{S}^n$ is the solution of the RDE
corresponding to $(Q,S,R,F)$ with terminal constraint $P(T)=F$. 


The analysis aims to bound the final position of the robot arm in the
presence of the disturbances $d$ and uncertainty at the joint
connecting the two arms.  Figure~\ref{fig:UncRobot} shows a block
diagram for the uncertain, linearized robot arm dynamics.
$\Delta \in \RH$ is an LTI uncertainty with $\|\Delta\|_\infty \le 1$.
The factors of $\sqrt{0.8}$ are included so that the overall level of
uncertainty at the joint is $0.8$. The error signal $e$
contains the two linearized joint angles:
\begin{align}
    e(t) & = \bmtx 1 & 0 & 0 & 0 \\ 0 & 0 & 1 & 0 \emtx x(t) 
         := C x(t)
\end{align} 

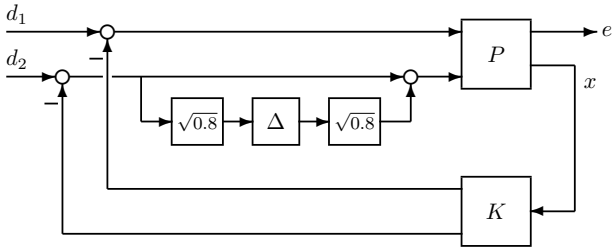
\begin{figure}[h]
\centering
\scalebox{0.85}{
  \begin{picture}(270,100)(0,-95)
    \thicklines
    \put(0,0){\vector(1,0){42}}
    \put(0,5){$d_1$}
    \put(0,-20){\vector(1,0){22}}
    \put(0,-15){$d_2$}
    \put(48,0){\vector(1,0){155}}
    \put(28,-20){\line(1,0){15}}
    \put(47,-20){\vector(1,0){130}}
    \put(60,-20){\line(0,-1){20}}
    \put(60,-40){\vector(1,0){14}}
    \put(74,-50){\framebox(22,20){\footnotesize $\sqrt{0.8}$}}
    \put(96,-40){\vector(1,0){14}}
    \put(110,-50){\framebox(20,20){$\Delta$}}
    \put(130,-40){\vector(1,0){14}}
    \put(144,-50){\framebox(22,20){\footnotesize $\sqrt{0.8}$}}
    \put(166,-40){\line(1,0){14}}
    \put(180,-40){\vector(0,1){17}}
    \put(180,-20){\circle{6}}
    \put(183,-20){\vector(1,0){20}}
    \put(203,-25){\framebox(30,30){$P$}}
    \put(233,0){\vector(1,0){30}}
    \put(265,-2){$e$}
    \put(257,-25){$x$}
    \put(233,-15){\line(1,0){20}}
    \put(253,-15){\line(0,-1){65}}
    \put(253,-80){\vector(-1,0){20}}
    \put(203,-95){\framebox(30,30){$K$}}
    \put(203,-70){\line(-1,0){158}}
    \put(45,-70){\vector(0,1){67}}
    \put(45,0){\circle{6}}
    \put(37,-12){\line(1,0){6}}
    \put(203,-90){\line(-1,0){178}}
    \put(25,-90){\vector(0,1){67}}
    \put(25,-20){\circle{6}}
    \put(17,-32){\line(1,0){6}}
\end{picture}
} 
\caption{Uncertain LTV Model for Two-Arm Robot}
\label{fig:UncRobot}
\end{figure}

Algorithm~\ref{alg:comb} was used to compute bounds on the robust
$\mathcal{L}_2$-to-Euclidean gain from $d$ to $e$ over the $T=5$sec
trajectory. The IQC is parameterized as in
Example~\ref{ex:LTIuncParam} with $v=1$ and $p=10$.
Algorithm~\ref{alg:comb} is initialized with $tol = 5 \times 10^{-3}$,
$N_{iter} = 10$, $t_{DLMI}$ as 20 evenly spaced points in $[0,T]$, and
$\tau_{sp}$ as 10 evenly spaced points in $[0,T]$. The algorithm
terminated after 3 iterations with a robust gain of
$\gamma_{CL}=0.0592$.  It took 103sec to perform this computation. For
comparison, the open-loop robust gain (with $K=0$) is
$\gamma_{OL}=941.6$. This computation terminated in 7 iterations and
took 321 sec. As expected, the feedback significantly reduces the
gain.


The results were tested by randomly generating 100 instances of
$\Delta$ with 0 to 6 states. Each instance of $\Delta$ was substituted
into Figure~\ref{fig:UncRobot} to generate a (nominal) LTV
closed-loop. The linearized closed-loop for each $\Delta$ was
simulated with disturbances $d$ such that
$\|d\|_{2,[0,T]} \leq \beta = 5$.  Figure~\ref{fig:etaGusims} shows
the linearized simulation results superimposed on the trim trajectory
$\bar \eta$. The final outputs $e(T)$ are designated by the white
dots. The light blue circle corresponds to
$\|e(T)\|_2^2 = \theta_1(T)^2 + \theta_2(T)^2 \leq \gamma_{CL}^2
\beta^2$.  As expected the simulated trajectories terminate in the
computed bound (cyan circle).  

Next, the closed-loop gain was evaluated for each $\Delta$ via
bisection with the (nominal) RDE. The largest gain was $0.0577$
achieved with the following uncertainty:
\begin{align*} 
  \Delta_{wc}(s) = 
    \frac{-0.7861 s^2 - 3.383 s - 3.631}{0.8 s^2 + 3.414 s + 3.631}.
\end{align*}
A worst-case disturbance was constructed for the closed-loop with
$\Delta_{wc}$. This construction is based on the two-point boundary
value problem that connects the performance to the RDE condition
(Lemma~\ref{lem:TPBVP} in Appendix~\ref{sec:BRLproof}). This yields a
trajectory with terminal condition very near to the boundary of the
cyan disk (see zoomed inset) indicating that the computed robustness
bounds are not overly conservative. Figure~\ref{fig:xyrobustGsims}
shows the same trajectories and bound but transformed to the Cartesian
space of the robot arm.


\begin{figure}[h!] 
  \centering
  \includegraphics[scale=0.27]{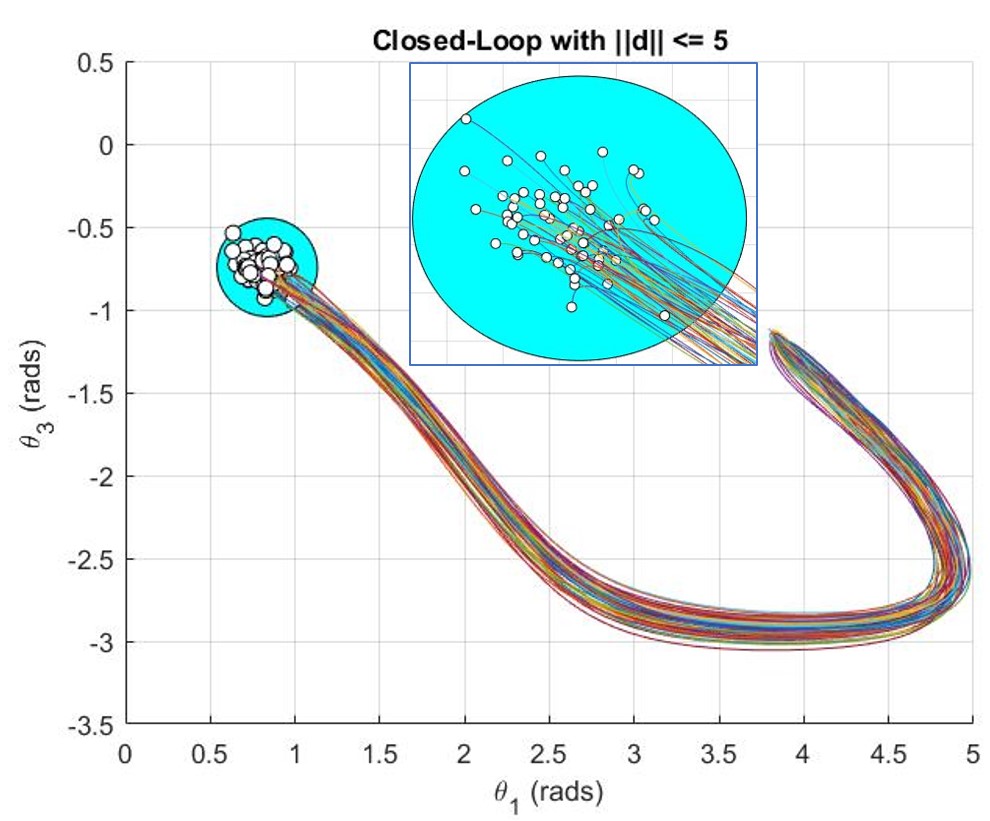}
  \caption{Closed-loop trajectories in the $(\theta_1, \theta_2)$
    space with $\Delta_{wc}$ and random disturbances
    $\|d\|_{2,[0,T]} \leq 5$. The robust bound on $e(T)$ is also shown
    (cyan circle).}
  \label{fig:etaGusims}
\end{figure}

\begin{figure}[h!] 
  \centering
  \includegraphics[scale=0.28]{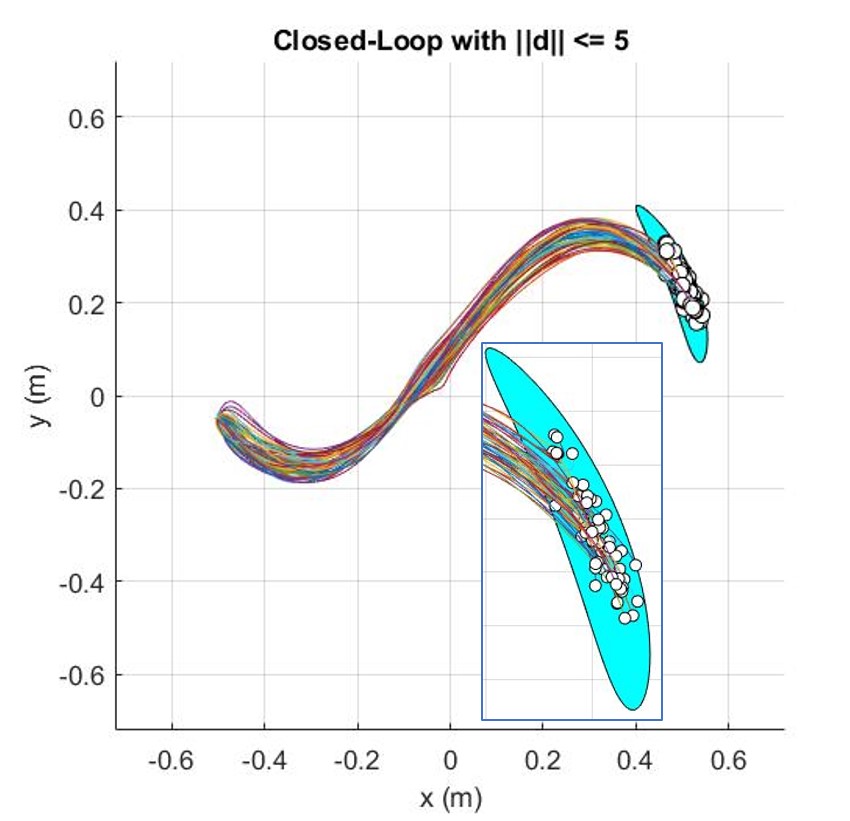}
  \caption{Closed-loop trajectories in Cartesian coordinates with
    $\Delta_{wc}$ and random disturbances $\|d\|_{2,[0,T]} \leq
    5$. The robust bound on $e(T)$ is also shown (cyan circle).}
  \label{fig:xyrobustGsims}
\end{figure}

\section{Conclusions}

This paper presented robust performance measures for the analysis of
uncertain LTV systems over a finite-horizon. The proposed numerical
algorithm combines differential linear matrix inequalities and Riccati
differential equations.  The utility of robust gains was demonstrated
with examples including a two-link robot arm. Future work will include
refinements to the algorithm along with methods to construct
worst-case perturbations.

\section*{Acknowledgments}
\label{sec:ack}

The authors gratefully acknowledge support from the National Science
Foundation under grants ECCS-1405413 and NSF-CMMI-1254129. A.
Packard acknowledges the generous support from the FANUC Corporation.

\begin{small}
\bibliographystyle{abbrv}
\bibliography{ltvbib}
\end{small}

\appendix
\section{Proof of Theorem~\ref{thm:BRL}}
\label{sec:BRLproof}

Theorem~\ref{thm:BRL} states an equivalence between: 1) a bound on the
quadratic cost $J$, 2) the existence of a solution $Y$ to a RDE, and
3) the existence of a solution $P$ to a RDI.  The proof of
$(3\Rightarrow 1)$ is given in the main text and the rest
of the proof is in this appendix. Section~\ref{sec:TPBVP} discusses a
related two-point boundary value problem (TPBVP).  The remainder of
the appendix provides proofs for $(1\Rightarrow 2)$,
$(2\Rightarrow 1)$, and $(1\Rightarrow 3)$.  This demonstrates
$(1\Leftrightarrow 2)$ and $(1\Leftrightarrow 3)$.  The equivalence
$(2\Leftrightarrow 3)$ follows from these results.  The TPBVP lemma
and proof of $(1\Rightarrow 2)$ is similar to the presentation given
in Section 3.7.4 of \cite{green95} for the special case of
finite-horizon induced $\mathcal{L}_2$ gains.

\subsection{Two-Point Boundary Value Problem}
\label{sec:TPBVP}

The LTV dynamics (Equation~\ref{eq:LTV1}) and quadratic cost $J$ are
defined by $(A,B)$ and $(Q,S,R,F)$, respectively.                 
Define a time-varying Hamiltonian $H:[0,T] \rightarrow
\Sm^{2n_x}$ as:
\begin{align}
  \label{eq:Ham}
  H:= \bmtx A & 0 \\ -Q & -A^T \emtx + \bmtx -B \\ S \emtx
  R^{-1} \bmtx S^T & B^T \emtx
\end{align}
A two-point boundary value problem (TPBVP) is defined for $t_0\in
[0,T]$ as:
\begin{align}
  \label{eq:TPBVP}
  \bmtx \dot{x}^*(t) \\ \dot{\lambda}(t) \emtx & = 
  H(t) \bmtx x^*(t) \\ \lambda(t) \emtx \\
  \label{eq:TPBVP_BC}
  \bmtx x^*(t_0) \\ \lambda(T) \emtx & = \bmtx 0 \\ F x^*(T) \emtx
\end{align}
Note that $x^*\equiv 0$ and $\lambda \equiv 0$ is a trivial solution
for this TPBVP.  The times $(t_0,T)$ are \emph{conjugate points} if
this TPBVP has a non-trivial solution.

\begin{lem}
  \label{lem:TPBVP}
  Let $t_0 \in [0,T]$ be given and assume $\exists \epsilon >0$ such
  that $J(d) \le -\epsilon \|d\|_{2,[0,T]}^2$
  $\forall d \in L_2[0,T]$. Then $(t_0,T)$ are not conjugate points
  of the TPBVP.
\end{lem}
\begin{proof}
  For $t_0=T$ the boundary conditions immediately imply $x^*(T)=\lambda(T)=0$, 
  i.e. $(t_0,T)$ are not conjugate points.  Thus assume
  $t_0 \in [0,T)$ and let $(x^*,\lambda)$ be any solution to the
  TPBVP. Define the signal:
  \begin{align*}
    \bar{d}(t) := \left\{ \begin{array}{cc}
            0 & t\le t_0 \\
           -R^{-1}(t) (S(t)^T x^*(t) + B(t)^T \lambda(t)) & t>t_0
                          \end{array}    \right.
  \end{align*}
  The TPBVP dynamics can be re-written in terms of $\bar{d}$:
  \begin{align}
    \label{eq:TPBVPdbar}
    \bmtx \dot{x}^*(t) \\ \dot{\lambda}(t) \emtx = 
       \bmtx A & 0 \\ -Q & -A^T \emtx 
        \bmtx x^*(t) \\ \lambda(t) \emtx
     - \bmtx -B \\ S \emtx \bar{d}(t)
  \end{align}
  Moreover, the response of the LTV system (Equation~\ref{eq:LTV1})
  with input $\bar{d}$ and initial condition $x(0)=0$ is given by
  $x(t)=0$ for $t< t_0$ and $x(t)=x^*(t)$ for $t \ge t_0$.

  The cost associated with the input $\bar{d}$ is:
  \begin{align*}
    J(\bar{d}) & = x^*(T)^T F x^*(T) \\
    & + \int_{t_0}^T \bsmtx x^*(t) \\ \bar d(t) \esmtx^T
    \bsmtx Q(t) & S(t) \\ S(t)^T & R(t) \esmtx
     \bsmtx x^*(t) \\ \bar d(t) \esmtx \, dt 
  \end{align*}
  The integrand can be simplified using the definition of $\bar{d}$
  and the TPBVP dynamics in Equation~\ref{eq:TPBVPdbar}:
  \begin{align*}
    \bsmtx x^* \\ \bar d \esmtx^T \bsmtx Q & S \\ S^T & R \esmtx
    \bsmtx x^* \\ \bar d \esmtx
    & = x^{*^T} \left( Q x^* + S \bar{d} \right) 
       + \bar{d}^T \left( S^T x^* + R \bar{d} \right) \\
    & = -x^{*^T} \left( \dot{\lambda} + A^T \lambda \right) 
      - \left( \dot{x}^* -Ax^* \right)^T \lambda \\
    & = -\frac{d}{dt} \left( x^{*^T} \lambda \right)
  \end{align*}
  These simplifications allow the cost to be rewritten as:
  \begin{align*}
    J(\bar{d}) = x^*(T)^T F x^*(T)
      - \int_{t_0}^T \frac{d}{dt} \left( x^{*^T}(t) \lambda(t) \right) \, dt      
  \end{align*}
  Integrate the last term and apply the boundary conditions
  $x^*(t_0)=0$ and $\lambda(T)=Fx^*(T)$ to show $J(\bar{d})=0$.  It is
  assumed that $J(\bar{d}) \le -\epsilon \| \bar{d}\|_{2,[0,T]}^2$ and
  hence $\bar{d}=0$.  Thus the TPBVP dynamics simplify to:
  \begin{align*}
    \bmtx \dot{x}^*(t) \\ \dot{\lambda}(t) \emtx = 
       \bmtx A & 0 \\ -Q & -A^T \emtx 
        \bmtx x^*(t) \\ \lambda(t) \emtx
  \end{align*}
  The boundary condition $x^*(t_0)=0$ thus implies $x^*\equiv 0$.
  This further implies that $\dot{\lambda} = -A^T \lambda$ with
  $\lambda(T)=Fx^*(T)=0$. Hence $\lambda \equiv 0$.  Therefore the
  TPBVP solution is trivial and $(t_0,T)$ are not conjugate
  points.
\end{proof}

\subsection{Proof of $\mathbf{(1\Rightarrow 2)}$ }


Assume $J(d) \le -\epsilon \|d\|_{2,[0,T]}^2$ $\forall d \in
\Ltwo$. Let $\Phi(t,T)$ denote the transition matrix associated with
the Hamiltonian dynamics (Equation~\ref{eq:TPBVP}) so that
for any $t \in [0,T]$:
\begin{align}
  \label{eq:HamPhi}
  \bmtx x^*(t) \\ \lambda(t) \emtx 
  = \Phi(t,T) \bmtx x^*(T) \\ \lambda(T) \emtx 
\end{align}
Note that a solution of the TPBVP must also satisfy the boundary
conditions in Equation~\ref{eq:TPBVP_BC}. Next define the following
matrix function:
\begin{align}
  \label{eq:X1X2}
  \bmtx X_1(t,T) \\ X_2(t,T) \emtx := \Phi(t,T) \bmtx I \\ F  \emtx 
\end{align}
Both $X_1$ and $X_2$ have $n_x$ rows compatible with
$\bsmtx x^*(t) \\ \lambda(t) \esmtx$.

It can be shown that $X_1(t,T)$ is nonsingular for all $t \in [0,T]$.
In particular, assume there exists a vector $v$ and time $t_0\in[0,T]$
such that $X_1(t_0,T)v=0$.  Set $x^*(T)=v$ and $\lambda(T)=Fv$.  The
state transition matrix (Equation~\ref{eq:HamPhi}) gives a solution
$(x^*,\lambda)$ for the Hamiltonian dynamics on $[0,T]$.  From the
definition $X_1$ it follows that $x^*(t_0) = X_1(t_0,T)v = 0$. Hence
$(x^*, \lambda)$ satisfy the TPBVP boundary conditions at $(t_0,T)$.
By Lemma~\ref{lem:TPBVP}, $J(d) \le -\epsilon \|d\|_{2,[0,T]}^2$
implies that $(t_0,T)$ are not conjugate points, i.e. the solution to
the TPBVP is trivial. Thus, $v=x^*(T)=0$ so that $X_1(t_0,T)$ is
nonsingular.

Finally, it can be verified that $Y(t):=X_2(t,T) X_1(t,T)^{-1}$
satisfies the RDE and $Y(T)=F$. It follows from $\Phi(T,T)=I$ and
Equation~\ref{eq:X1X2} that $X_1(T,T)=I$ and $X_2(T,T)=F$. Hence
$Y(T)=F$. Next, differentiating $Y(t)$ with respect to time $t$
yields:
\begin{align}
  \label{eq:Ydot1}
  \begin{split}
    \dot{Y} & = \dot{X}_2 X_1^{-1} - X_2 X_1^{-1} \dot{X}_1 X_1^{-1} \\
    & = \bmtx -Y & I \emtx \bmtx \dot{X}_1 \\ \dot{X}_2 \emtx X_1^{-1}
  \end{split}
\end{align}
By the definition of $X_1$ and $X_2$ in Equation~\ref{eq:X1X2},
\begin{align}
  \label{eq:Xdot}
    \bmtx \dot{X}_1 \\ \dot{X}_2 \emtx & = \dot{\Phi} \bmtx I \\ F \emtx 
    = H \Phi \bmtx I \\ F \emtx 
    = H \bmtx X_1 \\ X_2 \emtx
\end{align}
The second equality follows because $\Phi$ is the state transition
matrix for $H$. The third equality follows from
the definition of $X_1$ and $X_2$. Combine
Equations~\ref{eq:Ydot1} and \ref{eq:Xdot}:
\begin{align}
  \dot{Y} = \bmtx -Y & I \emtx H \bmtx I \\ Y \emtx
\end{align}
Substitute for $H$ (Eq.~\ref{eq:Ham}) to verify $Y$
solves the RDE.

\subsection{Proof of $\mathbf{(2\Rightarrow 1)}$ }

Assume the RDE has a solution $Y$.  The boundary conditions $Y(T)=F$
and $x(0)=0$ imply that $J$ can be equivalently written as:
\begin{align*}
  J(d) & = \int_0^T \bsmtx x(t) \\ d(t) \esmtx^T 
  \bsmtx Q(t) & S(t) \\ S(t)^T & R(t) \esmtx \bsmtx x(t) \\ d(t) \esmtx  
  \, dt \\ 
  & + \int_0^T \frac{d}{dt}\left( x(t)^T Y(t) x(t) \right) \, dt 
\end{align*}
The integrand in the second term can be expanded as $\dot{x}^T Y x +
x^T\dot{Y} x + x^T Y \dot{x}$. Substitute for $\dot{x}$ using the
system dynamics (Equation~\ref{eq:LTV1}) and for $\dot{Y}$ using the
RDE. After some algebra, this yields the following simplified form:
\begin{align}
  J(d) = \int_0^T \left( d(t) - \bar{d}(t) \right)^T R(t)
  \left( d(t) - \bar{d}(t) \right) \, dt
\end{align}
where $\bar{d}:=-R^{-1} (YB+S)^T x$.  Thus the cost function can be
bounded as follows:
\begin{align}
  \label{eq:Jbnd1}
  J(d) \le \alpha \| d-\bar{d}\|_{2,[0,T]}^2
\end{align}
where $\alpha := \max_{t\in[0,T]} \lambda_{max}(R(t))<0$.  

Finally, define the LTV system $W$ with input $d$ and output
$d-\bar{d}$:
\begin{align}
  W:= \left[ \begin{array}{c|c} A & B \\ \hline R^{-1} (YB+S)^T & I_{n_d} 
    \end{array} \right]
\end{align}
This system is invertible since the feedthrough matrix $I_{n_d}$ is
nonsingular.  Hence $W^{-1}$ exists and has finite gain, i.e.
$\|d\|_{2,[0,T]} \le \beta \| d-\bar{d}\|_{2,[0,T]}$ for some $\beta
<\infty$. This further yields $\alpha \| d-\bar{d}\|_{2,[0,T]}^2 \le
-\epsilon \| d\|_{2,[0,T]}^2$ where $\epsilon:=
-\frac{\alpha}{\beta^2}>0$.  Combine this bound with
Equation~\ref{eq:Jbnd1} to conclude that $J(d) \le -\epsilon
\|d\|_{2,[0,T]}^2$ $\forall d \in \Ltwo$.

\subsection{Proof of $\mathbf{(1\Rightarrow 3)}$ }

Assume $J(d) \le -\epsilon \|d\|_{2,[0,T]}^2$ $\forall d \in \Ltwo$
with $J$ defined by $(Q,S,R,F)$.  As noted previously, the LTV
system \eqref{eq:LTV1} has finite gain from $d$ to $x$,
i.e. $\|x\|_{2,[0,T]} \le \beta \|d\|_{2,[0,T]}$ for some
$\beta<\infty$.  Hence, there exists $\tilde\epsilon>0$ such that
\begin{align}
  \label{eq:Jpert}
  J(d) \le -\tilde\epsilon \left( \|x\|_{2,[0,T]}^2 + \|d\|_{2,[0,T]}^2
  \right)  \,\, \forall d \in L_2[0,T]
\end{align}
Define the perturbed cost function $\tilde{J}(d)$ with
$(\tilde{Q},S,R,F)$ where $\tilde{Q}:=Q+\tilde{\epsilon} I_{n_x}$.
The bound in Equation~\ref{eq:Jpert} is equivalent to $\tilde{J}(d)
\le -\tilde \epsilon \|d\|_{2,[0,T]}^2$ $\forall d \in \Ltwo$.  By $(1
\Rightarrow 2)$, there exists $\tilde{Y}:[0,T] \rightarrow \Sm^n$ such
that $\tilde{Y}(T)=F$ and
\begin{align*}
  \dot{ \tilde{Y} } + A^T \tilde{Y} + \tilde{Y} A + \tilde{Q} - (
  \tilde{Y} B+S) R^{-1} ( \tilde{Y}B+S)^T = 0
\end{align*}
Substitute $\tilde{Q}:=Q+ \tilde{\epsilon} I_{n_x}$ to obtain:
\begin{align*}
  \dot{ \tilde{Y} } + A^T \tilde{Y} + \tilde{Y} A + Q - ( \tilde{Y}
  B+S) R^{-1} ( \tilde{Y}B+S)^T = -\tilde{\epsilon} I_{n_x}
\end{align*}
Thus $\tilde{Y}$ satisfies the boundary condition
$\tilde{Y}(T) \succeq F$ and satisfies the strict RDI defined with
$(Q,S,R,F)$.

\end{document}